\documentclass[letterpaper, 10 pt, conference]{ieeeconf}

\usepackage{cite}
\usepackage{amsmath,amssymb,amsfonts}
\usepackage{bm}
\usepackage{algorithmic}
\usepackage{graphicx}
\usepackage{textcomp}
\usepackage{xcolor,color}
\usepackage{hyperref}
\usepackage{tikz}
\usetikzlibrary{positioning}
\usepackage{algorithm,tabularx}
\usepackage{multirow}
\usepackage{caption}
\makeatletter
\let\MYcaption\@makecaption
\makeatother 

\usepackage[font=footnotesize]{subcaption}

\makeatletter
\let\@makecaption\MYcaption
\makeatother

\newtheorem{definition}{Definition}

\newtheorem{theorem}{Theorem}

\newtheorem{remark}{Remark}

\newtheorem{assumption}{Assumption}

\IEEEoverridecommandlockouts
\begin{document}

\title{\bf Distributed Adaptive Backstepping Control for Vehicular Platoons with Mismatched Disturbances Using Vector String Lyapunov Functions}

\author{Zihao Song, Shirantha Welikala, Panos J. Antsaklis and Hai Lin\thanks{This work was supported by the National Science Foundation under  Grant CNS-1830335 and Grant IIS-2007949. The authors are with the Department of Electrical Engineering, University of Notre Dame, Notre Dame, IN 46556 USA (e-mail: {\tt zsong2@nd.edu; wwelikal@nd.edu; pantsakl@nd.edu; hlin1@nd.edu}).}}
\maketitle
\thispagestyle{empty}

\begin{abstract}
    In this paper, we consider the problem of platooning control with  mismatched disturbances using the distributed adaptive backstepping method. The main challenges are: (1) maintaining the compositionality and the distributed nature of the controller, and (2) ensuring the robustness of the controller with respect to general types of disturbances. To address these challenges, we first propose a novel notion that we named \emph{Vector String Lyapunov Function}, whose existence implies $l_2$ weak string stability. This notion is based on the vector Lyapunov function-based stability analysis, which depends on the input-to-state-stability of a comparison system. Using this notion, we propose an adaptive backstepping controller for the platoon such that the compositionality and the distributed nature of the controller can be ensured while the internal stability and string stability of the closed-loop system are formally guaranteed.  
    Finally, simulation examples are provided to illustrate the effectiveness of the proposed control algorithm. In particular, we provide simulation results comparing our proposed control algorithm with a recently proposed control algorithm from the literature, under two types of information flow topologies and disturbances.
\end{abstract}

\section{Introduction}\label{sec:intro}

A vehicular platoon is a promising transportation pattern and has become an inexorable future trend in modern transportation systems \cite{jia2015survey}. In this way, vehicles are arranged in a line and maneuvered as an entire team. This can potentially improve traffic mobility, fuel efficiency and travel safety \cite{badnava2021platoon}. Therefore, the control of platoons has attracted widespread attentions over the years.

From the perspective of control theory, earlier efforts for platooning control mainly focused on linear control methods, such as PID \cite{dasgupta2017merging,fiengo2019distributed}, LQR/LQG \cite{tavan2015optimal,wang2022optimal} and $\mathcal{H}_{\infty}$ controllers \cite{herman2014nonzero,gao2016robust,zheng2017platooning} to regulate the inter-vehicle distances. Recent years have seen the adoption of more advanced nonlinear control approaches, such as model predictive control (MPC) \cite{goli2019mpc,chen2018robust}, optimal control \cite{morbidi2013decentralized,ploeg2013controller}, consensus-based control \cite{syed2012coordinated}, sliding-mode control (SMC) \cite{xiao2011practical,guo2016distributed}, backstepping control \cite{zhu2018distributed,chou2019backstepping} and intelligent control (e.g., using neural networks (NN) \cite{ji2018adaptive} and fuzzy logic \cite{li2010design}).

Among the above approaches, the backstepping approach allows systematic controller design for nonlinear systems in a strict feedback form, which is immensely suitable for the control of the platoon system. Despite of this, backstepping control methods have not been paid enough attentions compared to other control methods like MPC, SMC and intelligent control due to their lack of compositionality with respect to the length of the platoon. In particular, this implies that the controllers need to be redesigned when the vehicles join/leave the platoon, which is a waste of time and expense.


To guarantee the compositionality, a robust finite-time backstepping controller was proposed in \cite{gao2021adaptive} for the platooning control. However, in there, each vehicle only tracks the state profile generated by the direct predecessor
and thus only the information from the predecessor is used, i.e., the predecessor following information flow topology (PF) is observed. Similarly, by considering the PF, an adaptive backstepping controller is presented in \cite{zhu2020v2v}. However, to further improve the control performance and provide more redundancies for a platoon, more information exchanges are required between the vehicles via different types of information flow topologies \cite{studli2017vehicular}. Therefore, these approaches are not feasible for distributed controller design. To fill this gap, as an extension of \cite{zhu2020v2v}, a distributed adaptive backstepping controller is designed in \cite{zhu2018distributed} limiting to symmetric-double-nearest-neighbor type information flow topology. Nevertheless, this controller is not compositional since its complexity grows exponentially with the platoon length. A recent trend for the control of vehicular platoons using distributed backstepping is to combine the controller with NN to estimate the uncertainties of the unmodeled dynamics \cite{liu2021distributed}. Although such a controller is compositional and distributed, the design and the tuning of the NNs are always challenging.

Besides, it is worth noting that none of the aforementioned works considered the mismatched disturbances. By mismatched disturbances, we mean the disturbances that enter the vehicular acceleration channel where the control input is not directly involved. This type of disturbances are actually very common in vehicular networks. For example, disturbances due to rough road conditions will enter the vehicle model through the equation of acceleration rather than that of the jerk. When this type of disturbances exist, spacing errors between vehicles may be amplified. To solve this problem, existing works mainly resort to disturbance observers \cite{wang2020finitetime,wang2019integral}, NN estimation \cite{an2022distributed,guo2019neuroadaptive}, or other feedback control methods (e.g., an integrated control law \cite{hu2022distributed}), and combine such mechanisms with different control methods that we mentioned earlier. Nevertheless, to the best of our knowledge, the combination of backstepping and adaptive methods for platooning control with mismatched disturbances has not been investigated.

Note also that the string stability should always be taken into account in platooning control. To this end, we use the $l_2$ weak string stability notion introduced in \cite{knorn2014passivity}.
Compared with other types of string stability, the $l_2$ weak string stability is defined in time domain and it contains the effect of disturbances  \cite{studli2017vehicular}. Besides, it relaxes the constraint that the errors are not allowed to increase along the platoon and has a closer relationship to Lyapunov-based analysis \cite{feng2019string}. 

Traditionally, distributed backstepping controller synthesis for vehicular platoons relies on scalar Lyapunov functions and canceling out complex coupling terms during the controller design. Thus, the resulting controllers are either not compositional \cite{zhu2018distributed} or not robust to disturbances \cite{liu2021distributed}. Inspired by the extension of Lyapunov functions to vector Lyapunov functions used for stability analysis of high-dimensional systems \cite{nersesov2006stability}, in this paper, we propose a novel notion that we named \emph{Vector String Lyapunov Function (VSLF)}, whose existence implies $l_2$ weak string stability. This notion is based on the vector Lyapunov function-based stability analysis, which depends on the input-to-state-stability (ISS) of a comparison system. By applying VSLFs, for each step of backstepping, there is no need to cancel out the said complex coupling terms and thus the controller is more robust compared to scalar Lyapunov function-based designs. More importantly, based on this VSLF notion, the compositionality and the distributed nature of the proposed adaptive backstepping controller can be ensured in a centralized manner.

Therefore, based on the above discussion, it remains an open problem to design an adaptive backstepping controller for the vehicular platoons to ensure the compositional and distributed properties meanwhile guaranteeing the robustness with respect to general types of disturbances, e.g., mismatched disturbances. Our primary contributions can be summarized as follows:
\begin{enumerate}
    \item A novel distributed adaptive backstepping controller is designed for the control of the vehicular platoon with mismatched disturbances;
    \item The internal and string stability under the designed controller are established using a novel \emph{VSLF} notion;
    \item The proposed controller ensures the compositionality and the distributed nature with the aid of centralized adaptive laws;
    \item The robustness of the proposed controller is illustrated using simulation experiments with two types of disturbances (sinusoidal and Gaussian), and by comparing the tracking performance with an existing work \cite{wang2020finitetime}.
\end{enumerate}

This paper is organized as follows. Some necessary preliminaries and the problem formulation are presented in Section \ref{sec:background}. Our main results are presented in Section \ref{sec:main_results}, and are supported by {multiple simulation examples in Section \ref{sec:simulation}. Concluding remarks are provided in Section \ref{sec:conclusion}.

\section{Background}\label{sec:background}

\subsection{Notations}

The sets of real, natural, positive real and non-negative real numbers are denoted by $\mathbb{R}$, $\mathbb{N}$, $\mathbb{R}_+$ and $\mathbb{R}_{\geq 0}$, respectively. $\mathbb{R}^{n\times m}$ denotes the vector space of real matrices with $n$ rows and $m$ columns. An $n$-dimensional real vector is denoted by $\mathbb{R}^n$. Define the index sets $\mathcal{I}_N:=\{1, 2,...,N\}$ and $\mathcal{I}_N^0:=\mathcal{I}_N\cup\{0\}$, where $N\in \mathbb{N}$. For two vectors $\mathbf{x}$, $\mathbf{y}\in\mathbb{R}^n$, we use $\mathbf{x}\preceq \mathbf{y}$ to indicate that every component of $\mathbf{x}-\mathbf{y}$ is non-positive, i.e., $x_i\leq y_i$ for all $i\in \mathcal{I}_N$. $\mathbf{x}\ngeq \mathbf{y}$ denotes that at least one element in $\mathbf{x}$, e.g., $x_i$ satisfies $x_i<y_i$. The $1$-norm and Euclidean norm of a vector are given by $|\mathbf{x}|_1=\sum_{i=1}^n |x_i|$ and $|\mathbf{x}|_2 = |\mathbf{x}| = \sqrt{\mathbf{x}^T\mathbf{x}}$, respectively. The $\mathcal{L}_2$ and $\mathcal{L}_{\infty}$ vector function norms are given by $\|\mathbf{x}(t)\|=\sqrt{\int_{0}^{\infty}|\mathbf{x}(t)|^2dt}$ and $\|\mathbf{x}(t)\|_{\infty} = \max_{i\in\mathcal{I}_N} \{|x_i(t)|\}$, respectively. We use $\mathcal{K}$, $\mathcal{K}_{\infty}$ and $\mathcal{KL}$ to denote different classes of comparison functions, see e.g., \cite{sontag1995characterizations}. For a function of time $t$, we will omit the notation $(t)$ when it is not necessary for ease of expression.




\subsection{Preliminary Definitions}\label{subsec:RAS_for_PWA}

Consider the tracking error dynamics of a platoon $\Sigma$ comprised of $N$ followers $\Sigma_i, i\in\mathcal{I}_N$ and a leader $\Sigma_0$ as:
\begin{equation}\label{Eq:general_platoon}
        \Sigma:\left\{\begin{array}{ll}
    \dot{\mathbf{e}}=\mathcal{F}(\mathbf{e},\boldsymbol{\omega}), &  \\
    \mathbf{y}=\mathcal{G}(\mathbf{e}), & 
    \end{array}\right.
    \end{equation}
where $\mathbf{e}=[e_1,e_2,...,e_N]^T\in\mathbb{R}^{nN}$ is the tracking error of the platoon and $\boldsymbol{\omega}=[\omega_1,\omega_2,...,\omega_N]^T\in\mathbb{R}^{mN}$ is the external disturbance, where $e_i\in\mathbb{R}^n$ and $\omega_i\in\mathbb{R}^m$ are the tracking error and the external disturbance of the $i^{th}$ vehicle, respectively, for $i\in\mathcal{I}_N$. The function $\mathcal{F}:\mathbb{R}^{nN}\times\mathbb{R}^{mN}\rightarrow\mathbb{R}^{nN}$ satisfies $\mathcal{F}(0, 0) = 0$ and is assumed to be locally Lipschitz continuous in both arguments. The function $\mathcal{G}:\mathbb{R}^{nN}\rightarrow\mathbb{R}^{pN}$ is the output map.



For string stability analysis of the tracking error dynamics \eqref{Eq:general_platoon}, we first recall the definition of the $l_2$ weak string stability.
\begin{definition}(\textit{$l_2$ weak string stability \cite{knorn2014passivity}})\label{def:l2_weakly_ss}
    The equilibrium $\mathbf{e}=\mathbf{0}$ of \eqref{Eq:general_platoon} is $l_2$ weakly string stable with respect to the disturbances $\boldsymbol{\omega}(t)$, if given any $\epsilon$, there exists $\delta_1(\epsilon)>0$ and $\delta_2(\epsilon)>0$ (both independent of $N$) such that
\begin{equation}\label{Eq:DSS_initial}
    |\mathbf{e}(0)|<\delta_1(\epsilon),\ \|\boldsymbol{\omega}(\cdot)\|<\delta_2(\epsilon)
\end{equation}
implies:
\begin{equation}\label{Eq:DSS_estimation}
    \|\mathbf{e}(t)\|_{\infty}:=\underset{t\geq 0}{\sup}\ |\mathbf{e}(t)|<\epsilon,\ \mbox{$\forall N\geq 1$}.
\end{equation}
\end{definition}



Now, to maintain the compositionality and the distributed nature of the controller, we propose a novel notion that we named \emph{VSLF}, whose existence implies the $l_2$ weak string stability (formally proven in Section \ref{sec:main_results}).

\begin{definition}(\textit{VSLF})\label{def:VSLF}
    Assume that the tracking error dynamics of the platoon \eqref{Eq:general_platoon} can be partitioned into $\bar{N}$ subsystems, where $\bar{N}=n$. Consider a continuous vector function $\mathcal{V}=[\mathcal{V}_1,...,\mathcal{V}_{\bar{N}}]^T:\mathbb{R}^{nN}\rightarrow\mathbb{R}_{\geq 0}^{\bar{N}}$, with each $\mathcal{V}_k:\mathbb{R}^{N}\rightarrow\mathbb{R}_{\geq 0}$ being a scalar Lyapunov function of the $k^{th}$ subsystem of \eqref{Eq:general_platoon} for $k\in \mathcal{I}_{\bar{N}}$ that satisfies the following inequalities:
\begin{gather}
    \underline{\alpha}(|\mathbf{e}|)\leq \sum_{k=1}^{\bar{N}} \mu_k\mathcal{V}_k\leq\overline{\alpha}(|\mathbf{e}|) \label{def:state_error_bound},\\
    D^+\mathcal{V}\preceq \Gamma(\mathcal{V})+\gamma(\|\boldsymbol{\omega}\|)    \label{def:vector_dissipativity_condition},
\end{gather}
where $\underline{\alpha}$ and $\overline{\alpha}$ are some class-$\mathcal{K}_{\infty}$ functions and $\mu_k\in\mathbb{R}_+$; $D^+\mathcal{V}$ is the upper-right Dini derivative in the direction of $\mathcal{V}$; $\Gamma(\cdot)$ is some class-$\mathcal{W}$ function\footnote{A function $\Gamma(\mathbf{s})=[\Gamma_1(\mathbf{s}),...,\Gamma_{\bar{N}}(\mathbf{s})]^T:\mathbb{R}^{\bar{N}}\rightarrow\mathbb{R}^{\bar{N}}$ is of class-$\mathcal{W}$ if $\Gamma_i(\mathbf{s}')\leq\Gamma_i(\mathbf{s}'')$, $i\in\mathcal{I}_{\bar{N}}$ for all $\mathbf{s}',\mathbf{s}''\in\mathbb{R}^{\bar{N}}$ such that $s'_j\leq s''_j$, $s'_i=s''_i$, $j\in\mathcal{I}_{\bar{N}}$, $i\neq j$, where $s_i$ denotes the $i^{th}$ component of the vector $\mathbf{s}$. Note that if $\Gamma(\mathbf{s})=\Gamma\mathbf{s}$ (with a slight abuse of notation), then all the off-diagonal entries of the matrix $\Gamma$ are non-negative \cite{nersesov2006stability}.};
$\gamma(\cdot)$ is some vector function with non-negative components and $\gamma(\|\boldsymbol{\omega}\|)$ denotes $[\gamma_1(\|\boldsymbol{\omega}_1\|),...,\gamma_{\bar{N}}(\|\boldsymbol{\omega}_{\bar{N}}\|)]^T$. If there exists a comparison system for \eqref{def:vector_dissipativity_condition} of the form:
\begin{equation}\label{Eq:comparison_system}
\dot{\mathbf{z}}=\Gamma(\mathbf{z})+\gamma(\|\boldsymbol{\omega}\|),
\end{equation}
which is ISS from $\boldsymbol{\omega}$ to $\mathbf{z}$, where $\mathbf{z}\in\mathbb{R}_{\geq 0}^{\bar{N}}$, then $\mathcal{V}$ is called a VSLF of the platoon $\Sigma$.
\end{definition}
\begin{remark}
    Any platoon system $\Sigma$ can be partitioned into $\bar{N}$ subsystems since we can distinguish the tracking errors of spacing, velocity and acceleration for each vehicle and stack each type of errors into a vector. Then, we can simply refer to each type of tracking error as a subsystem.
\end{remark}
\begin{remark}
    To guarantee the ISS of comparison system in a general form as \eqref{Eq:comparison_system}, the small-gain condition with additive input is a sufficient condition (for more details, see Corollary 5.6 of \cite{dashkovskiy2010small}). For undirected information flow topology, such a small-gain condition takes the form $D\circ\Gamma(\mathbf{z})\ngeq \mathbf{z}$, where the operator $D:=\mbox{diag}(\mbox{id}+\alpha)$ ($\mbox{id}$ is the identity map) with $\alpha\in\mathcal{K}_{\infty}$. Note that for linear comparison systems, this condition is reduced to the condition that the linear system matrix $\Gamma$ being Hurwitz. 
\end{remark}


\subsection{Problem Formulation}

Consider the longitudinal dynamics of the $i^{th}$ follower $\Sigma_i$ in the platoon $\Sigma$ as in \cite{chou2019backstepping} but with mismatched disturbances:
\begin{equation}\label{Eq:platoon_longitudinal}
        \Sigma_i:
        \begin{cases}
            \dot{x}_{i}(t)\hspace{-1mm}=\hspace{-1mm}v_{i}(t), \\
            \dot{v}_{i}(t)\hspace{-1mm}=\hspace{-1mm}a_{i}(t)\hspace{-1mm}+\hspace{-1mm}d_{vi}(t), \\
            \dot{a}_{i}(t)\hspace{-1mm}=\hspace{-1mm}f_{i}(v_{i}(t),a_{i}(t))\hspace{-1mm}+\hspace{-1mm}g_{i}(v_{i}(t))u_i(t)\hspace{-1mm}+\hspace{-1mm}d_{ai}(t),
        \end{cases}
    \end{equation}
for $i\in\mathcal{I}_N$, where
\begin{equation}
    \begin{split}
        f_{i}(v_{i},a_{i})&:=-\frac{1}{\tau_i}\left(a_i+\frac{A_{f,i}\rho C_{d,i}v_i^2}{2m_i}+C_{r,i}\right)\\
        &-\frac{A_{f,i}\rho C_{d,i}v_ia_i}{m_i},\\
        g_{i}(v_{i})&:=\frac{1}{m_i\tau_i},
    \end{split}
\end{equation}
are the nonlinear aerodynamics of the vehicle. We lumped all the uncertainties in these two terms together with the external disturbances in the jerk as $d_{ai}$, and $d_{vi}$ is the mismatched disturbance. $x_i\in\mathbb{R}$, $v_i\in\mathbb{R}$ and $a_i\in\mathbb{R}$ are the position, speed and acceleration, respectively; $m_i$ is the mass; $A_{f,i}$ is the effective frontal area; $\rho$ is the air density; $C_{d,i}$ is the coefficient of the aerodynamic drag; $C_{r,i}$ is the coefficient of the rolling resistance; $\tau_i$ is the engine time constant; and $u_i$ is the control input to be designed. 
\begin{assumption}\label{ass:1}
The disturbance terms $d_{vi}$ and $d_{ai}$ in \eqref{Eq:platoon_longitudinal} and its first time derivative are bounded by known constants, i.e., $|d_{vi}|\leq\delta_{vi}$, $|d_{ai}|\leq\delta_{ai}$, $|\dot{d}_{vi}|\leq\bar{\delta}_{vi}$ and $|\dot{d}_{ai}|\leq\bar{\delta}_{ai}$, where $\delta_{vi}$, $\delta_{ai}$, $\bar{\delta}_{vi}$, $\bar{\delta}_{ai}\in\mathbb{R}_+$.
\end{assumption}

Regarding the dynamics of the leading vehicle $\Sigma_0$, we make the following assumption.
\begin{assumption}\label{ass:2}
    The dynamics of the leading vehicle are
\begin{equation}\label{Eq:leader_dynamics}
         \Sigma_0: \left \{\begin{array}{ll}
        \dot{x}_0(t)=v_0(t), &\\
        \dot{v}_0(t)=a_0(t),
        \end{array}\right.
    \end{equation}
where $x_0\in\mathbb{R}$, $v_0\in\mathbb{R}$ and $a_0\in\mathbb{R}$ are the position, speed and acceleration, respectively. The velocity profile of the leader $v_0(t)$ is free to be selected such that the corresponding jerk $\dot{a}_0(t)=0$ \cite{zhu2018distributed}.
\end{assumption}

In particular, the information flow topology of the platoon is modeled by an augmented directed graph $G(V,E)$, where $V\in\mathcal{I}_N^0$ is the node set, and $E\subseteq V\times V$ is the edge set. To characterize the connectivity of the graph $G$, we denote the adjacency matrix as $\mathcal{A}=[A_{ij}]\in\mathbb{R}^{N\times N}$, where $A_{ij}=1$ if $(j,i)\in E$ and $A_{ij}=0$ otherwise, for $i,j\in\mathcal{I}_N$. Note that $(j,i)\in E$ implies that there is an information flow from node $j$ to node $i$, i.e., node $i$ receives the information of node $j$. In the graph $G$, we omit self-loops, i.e., $A_{ii}=0$, for $i\in\mathcal{I}_N$. The Laplacian matrix of the graph $G$ is defined as $\mathcal{L}=[L_{ij}]\in\mathbb{R}^{N\times N}$, where $L_{ij}=-A_{ij}$ if $i\neq j$, and $L_{ij}=\sum_{k=1,k\neq i}^N A_{ik}$, otherwise, for $i,j\in\mathcal{I}_N$. The pinning matrix of the graph $G$ is defined as $\mathcal{P}=\mbox{diag}\{P_i\}\in\mathbb{R}^{N\times N}$, where $P_i$ is used to indicate the existence of an edge from the leader node $0$ to node $i$, i.e., if $P_i=1$, the $i^{th}$ node can receive the leader’s information; $P_i=0$, otherwise.

\begin{assumption}(\textit{positive definite topology \cite{wu2016distributed}})\label{ass:3}
    There exists a spanning tree in the graph $G$ and the information flow topology between followers is assumed to be undirected.
\end{assumption}

To achieve the coordinated control of the vehicular platoon, the synchronization error for the $i^{th}$ follower can be defined as follows:
\begin{align*}
    &e_{ix}:=P_i(x_0-x_i-d_i)+\sum_{j\in\mathcal{I}_N\backslash\{i\}}A_{ij}(x_j-x_i-d_i+d_j), \\
    &e_{iv}:=\dot{e}_{ix}=P_i(v_0-v_i)+\sum_{j\in\mathcal{I}_N\backslash\{i\}}A_{ij}(v_j-v_i),
\end{align*}
where $d_m=\sum_{k\in\mathcal{I}_m}(\delta_{dk}+L_k)$, for $m=i,j$. Here, $\delta_{dk}$ is the desired inter-distance between two successive vehicles and $L_k$ is the length of the $k^{th}$ vehicle as shown in Fig. \ref{fig:config_platoon}.

\begin{figure}[!t]
    \centering
    \includegraphics[width=0.95\linewidth]{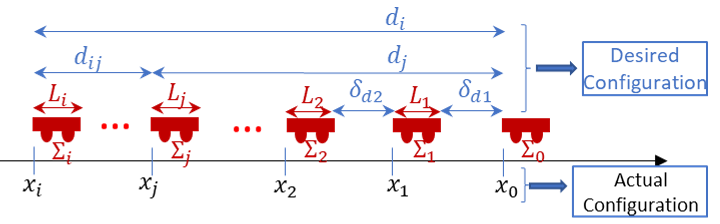}
    \caption{Configuration of the platoon.}
    \label{fig:config_platoon}
\end{figure}

Thus, the error dynamics of the $i^{th}$ follower $\Sigma_i$ in the platoon can be written as
\begin{equation}\label{Eq:platoon_error}
        \Sigma_{ei}: \left \{\begin{array}{ll}
        \dot{e}_{ix}=e_{iv}, &\\
        \dot{e}_{iv}=P_i(a_0-a_i)+\sum_{j=1,j\neq i}^{N}A_{ij}(a_j-a_i)+\bar{d}_{vi}, &\\
        \dot{e}_{ai}=-f_{i}(v_{i},a_{i})-g_{i}(v_{i})u_i-d_{ai}, &
        \end{array}\right.
    \end{equation}
where $\bar{d}_{vi}=-P_id_{vi}+\sum_{j=1,j\neq i}^{N}A_{ij}(d_{vj}-d_{vi})$.

The stacked (vectored) version of \eqref{Eq:platoon_error} can be written as: 
\begin{equation}\label{Eq:platoon_error_vector}
        \Sigma_{e}: \left \{\begin{array}{ll}
        \dot{\mathbf{e}}_x=\mathbf{e}_v, &\\
        \dot{\mathbf{e}}_v=\mathcal{H}(\mathbf{e}_a + \mathbf{D}_v), &\\
        \dot{\mathbf{e}}_a=-\mathbf{F}(\mathbf{v},\mathbf{a})-\mathbf{G}(\mathbf{v})\mathbf{u}+\mathbf{D}_a, &
        \end{array}\right.
    \end{equation}
where the matrix $\mathcal{H}=\mathcal{P+L}$, the nonlinear vector function $\mathbf{F}(\mathbf{v},\mathbf{a})=[f_1(v_1,a_1),...,f_{N}(v_{N},a_{N})]^T$, $\mathbf{G}(\mathbf{v})=\mbox{diag}\{g_{i}(v_{i})\}$. The mismatched disturbances can be lumped as $\mathbf{D}_v=-[d_{v1},d_{v2},...,d_{vN}]^T$. The disturbances in the input channel can be lumped as $\mathbf{D}_a=-[d_{a1},d_{a2},...,d_{aN}]^T$. Based on the Assumption \ref{ass:1}, the disturbances $\mathbf{D}_v$, $\mathbf{D}_a$ and their first time derivatives are bounded, i.e., $\|\mathbf{D}_v\|\leq\Delta_v$, $\|\mathbf{D}_a\|\leq\Delta_a$, $\|\dot{\mathbf{D}}_v\|\leq\bar{\Delta}_v$, $\|\dot{\mathbf{D}}_a\|\leq\bar{\Delta}_a$, where $\Delta_v$, $\Delta_a$, $\bar{\Delta}_v$, $\bar{\Delta}_a$ are determined by the constants in Assumption \ref{ass:1}.

Our objective is to design an adaptive backstepping controller for the system \eqref{Eq:platoon_error_vector} using the proposed VSLF such that it is compositional, distributed and robust to general types of disturbances. Besides, the internal and string stability of the platoon $\Sigma$ are guaranteed.


\begin{remark}
By Assumption \ref{ass:3}, we mean that there is a directed path from the root node (leading vehicle) to every other node (following vehicles) in the graph \cite{wang2014leader}. Further, since the information flow topology is undirected between the followers, it is trivial to show that $\mathcal{H}=\mathcal{H}^T>0$ \cite{wu2016distributed}. Note also that Assumption 3 holds for information flow topology types such as bidirectional, bidirectional-leader, symmetric-double-nearest-neighbor, etc \cite{wu2016distributed}.
\end{remark}

\section{Main Results}\label{sec:main_results}

In this section, we provide our main theoretical results. First, we show that if such a proposed VSLF can be found, it implies $l_2$ weak string stability. Then, based on this claim, we design a novel adaptive backstepping controller such that it is compositional and distributed in a centralized manner. Finally, we show that the internal stability and the string stability of the platoon are both guaranteed using our proposed VSLF notion.


\subsection{The Existence of a VSLF Implies $l_2$ Weak String Stability}

In this part, we prove that the existence of a VSLF in Def. \ref{def:VSLF} implies the $l_2$ weak string stability in Def. \ref{def:l2_weakly_ss}.
\begin{theorem}\label{theorem:VSS_DSS}
    If there exists a VSLF for the tracking error dynamics \eqref{Eq:general_platoon} of the platoon $\Sigma$, then it is $l_2$ weakly string stable.
\end{theorem}
\begin{proof}
Based on Lemma 3.3 of \cite{ruffer2010connection}, we know that the solution of \eqref{Eq:comparison_system}, as long as it exists in the positive orthant, will evolve within this region. Then, based on the Def. \ref{def:VSLF} and the comparison principle, for $\mathbf{z}(0)=\mathcal{V}(0)$, we have:
\begin{equation}\label{Eq:Ve_z}
    \mathcal{V}(\mathbf{e}(t))\preceq\mathbf{z}(t), \mbox{ for all $t\geq 0$,}
\end{equation}
where $\mathbf{z}(t)$ is the solution of \eqref{Eq:comparison_system} with respect to $\mathbf{z}(0)$ and $\boldsymbol{\omega}(t)$. Since the comparison system \eqref{Eq:comparison_system} is ISS, it implies that:
\begin{equation}\label{Eq:comparison_ISS_condition}
    |\mathbf{z}(t)|\leq\beta_z(|\mathbf{z}(0)|,t)+\gamma_z(\|\boldsymbol{\omega}\|),
\end{equation}
where $\beta_z$ is a class-$\mathcal{KL}$ function and $\gamma_z$ is a class-$\mathcal{K}$ function.

Based on \eqref{Eq:Ve_z} and Cauchy–Schwarz inequality, we have: 
\begin{equation}\label{Eq:Cauchy_Schwarz}
    \sum_{k=1}^{\bar{N}} \mu_k\mathcal{V}_k=\boldsymbol{\mu}^T\mathcal{V}\leq\boldsymbol{\mu}^T\mathbf{z}\leq |\boldsymbol{\mu}||\mathbf{z}|,
\end{equation}
where $\boldsymbol{\mu}=[\mu_1,\mu_2,...,\mu_{\bar{N}}]^T$.

Therefore, using \eqref{Eq:comparison_ISS_condition} and \eqref{Eq:Cauchy_Schwarz}, we have:
\begin{align}
    \sum_{k=1}^{\bar{N}} \mu_k\mathcal{V}_k&\leq\bar{\beta}_z(|\mathbf{z}(0)|,t)+\bar{\gamma}_z(\|\boldsymbol{\omega}\|) \nonumber \\
    &\leq \bar{\beta}_z(|\mathbf{z}(0)|_1,t)+\bar{\gamma}_z(\|\boldsymbol{\omega}\|) \nonumber \\
    &\leq\bar{\beta}_z\left(\sum_{k=1}^{\bar{N}}\frac{\mu_k}{\mu_{\min}}z_k(0),t\right)+\bar{\gamma}_z(\|\boldsymbol{\omega}\|)  \nonumber  \\
    &=\bar{\beta}_z\left(\sum_{k=1}^{\bar{N}}\frac{\mu_k}{\mu_{\min}}\mathcal{V}_k(0),t\right)+\bar{\gamma}_z(\|\boldsymbol{\omega}\|),
\end{align}
where $\bar{\beta}_z(\cdot,t)=|\boldsymbol{\mu}|\beta_z(\cdot,t)$ and $\bar{\gamma}_z=|\boldsymbol{\mu}|\gamma_z(\cdot)$. 

Thus, based on condition \eqref{def:state_error_bound}, we have:
\begin{align}\label{Eq:ISSS_condition}
    &\|\mathbf{e}(t)\|_{\infty} = \max_{i\in\mathcal{I}_N}\{|e_i(t)|\}\leq
    \sqrt{\mathbf{e}^T(t)\mathbf{e}(t)} = |\mathbf{e}(t)| \nonumber \\
    \leq&\underline{\alpha}^{-1}\left(\bar{\beta}_z\left(\frac{1}{\mu_{\min}}\overline{\alpha}(|\mathbf{e}(0)|),t\right)+\bar{\gamma}_z(\|\boldsymbol{\omega}\|)\right)  \nonumber \\
    \leq&\underline{\alpha}^{-1}\left(2\bar{\beta}_{z0}\left(\frac{1}{\mu_{\min}}\overline{\alpha}\left(|\mathbf{e}(0)|\right)\right)\right)+\underline{\alpha}^{-1}\left(2\bar{\gamma}_z(\|\boldsymbol{\omega}\|)\right),
\end{align}
where $\bar{\beta}_{z0}(\cdot)=\bar{\beta}_z(\cdot,0)$, which is a class-$\mathcal{K}$ function.

To bound the right-hand side of \eqref{Eq:ISSS_condition} such that $\|\mathbf{e}(t)\|_{\infty}<\frac{1}{p}\epsilon+\frac{1}{q}\epsilon=\epsilon$, where $p$, $q$ satisfy $\frac{1}{p}+\frac{1}{q}=1$, we can find bounds for $|\mathbf{e}(0)|$ and $\|\boldsymbol{\omega}(t)\|$ as $\delta_1(\epsilon)=\overline{\alpha}^{-1}\left(\mu_{\min}\bar{\beta}_{z0}^{-1}\left(\frac{1}{2}\underline{\alpha}\left(\frac{1}{p}\epsilon\right)\right)\right)$ and $\delta_2(\epsilon)=\bar{\gamma}_z^{-1}\left(\frac{1}{2}\underline{\alpha}\left(\frac{1}{q}\epsilon\right)\right)$ satisfying \eqref{Eq:DSS_initial}, respectively. This completes the proof.
\end{proof}




\subsection{Distributed Adaptive Backstepping Controller Design}

In this part, we design our distributed adaptive backstepping controller for the platoon error dynamics \eqref{Eq:platoon_error_vector}.

\paragraph*{\textbf{Step 1}} Define $\mathbf{e}_1:=\mathbf{e}_x$, $\mathbf{e}_2:=\mathbf{e}_v-\mathbf{e}_v^*$, where $\mathbf{e}_v^*$ is some virtual controller that is designed to stabilize $\mathbf{e}_1$. Choose the first Lyapunov function as $V_1:=|\mathbf{e}_1|$, then:
\begin{equation}\label{Eq:dotV1}
    \dot{V}_1=\mbox{sgn}^T(\mathbf{e}_1)(\mathbf{e}_2+\mathbf{e}_v^*),
\end{equation}
where $\mbox{sgn}(\cdot):=\frac{(\cdot)}{|\cdot|}$. Design the virtual control law as:
\begin{equation}\label{Eq:1st_virtual_law}
    \mathbf{e}_v^*:=-\mathbf{K}_1\mathbf{e}_1,
\end{equation}
where $\mathbf{K}_1>0$ is a diagonal control parameter matrix.

Substitute \eqref{Eq:1st_virtual_law} into \eqref{Eq:dotV1} to obtain:
\begin{align}
  \hspace{-2mm}  \dot{V}_1=&\mbox{sgn}^T(\mathbf{e}_1)(\mathbf{e}_2-\mathbf{K}_1\mathbf{e}_1)    \nonumber    \\
  =&\frac{\mathbf{e}_1^T}{|\mathbf{e}_1|}(\mathbf{e}_2-\mathbf{K}_1\mathbf{e}_1)  \nonumber    \\
  \leq&\frac{1}{|\mathbf{e}_1|}|\mathbf{e}_1||\mathbf{e}_2|-\lambda_{\min}(\mathbf{K}_1)|\mathbf{e}_1|   \nonumber    \\
  \leq& -\lambda_{\min}(\mathbf{K}_1)|\mathbf{e}_1|+|\mathbf{e}_2|.
\end{align}

Then, the first equation of \eqref{Eq:platoon_error_vector} becomes:
\begin{equation}
    \dot{\mathbf{e}}_1=\dot{\mathbf{e}}_x=-\mathbf{K}_1\mathbf{e}_1+\mathbf{e}_2.
\end{equation}

\paragraph*{\textbf{Step 2}} Define $\mathbf{e}_3:=\mathbf{e}_a-\mathbf{e}_a^*$, where $\mathbf{e}_a^*$ is some virtual controller that is designed to stabilize $\mathbf{e}_2$. Thus, the error dynamics of $\mathbf{e}_2$ can be written as:
\begin{align}\label{Eq:error2}
    \dot{\mathbf{e}}_2&=\dot{\mathbf{e}}_v-\dot{\mathbf{e}}_v^*  \nonumber  \\
        &=\mathcal{H}(\mathbf{e}_a+\mathbf{D}_v)+\mathbf{K}_1\dot{\mathbf{e}}_1   \nonumber \\
        &=\mathcal{H}(\mathbf{e}_a+\mathbf{D}_v)+\mathbf{K}_1(-\mathbf{K}_1\mathbf{e}_1+\mathbf{e}_2)  \nonumber  \\
        &=\mathcal{H}(\mathbf{e}_3+\mathbf{e}_a^*+\mathbf{D}_v)+\mathbf{K}_1(-\mathbf{K}_1\mathbf{e}_1+\mathbf{e}_2) \nonumber   \\
        &=\mathcal{H}(\mathbf{e}_3+\mathbf{e}_a^*+\mathbf{D}_v)-\mathbf{K}_1^2\mathbf{e}_1+\mathbf{K}_1\mathbf{e}_2.
\end{align}

Design the virtual controller $\mathbf{e}_a^*$ as follows:
\begin{equation}\label{Eq:2nd_virtual_law}
    \mathbf{e}_a^*=-\mathbf{K}_2\mathbf{e}_2-\hat{\mathbf{D}}_v,
\end{equation}
where $\mathbf{K}_2>0$ is a diagonal control parameter matrix; $\hat{\mathbf{D}}_v$ is the result of some adaptive law to be designed. Then, by substituting \eqref{Eq:2nd_virtual_law} into \eqref{Eq:error2}, it becomes:
\begin{equation}
    \begin{split}
        \dot{\mathbf{e}}_2&=\mathcal{H}\mathbf{e}_3-\mathcal{H}\mathbf{K}_2\mathbf{e}_2-\mathbf{K}_1^2\mathbf{e}_1+\mathbf{K}_1\mathbf{e}_2+\mathcal{H}\tilde{\mathbf{D}}_v,
    \end{split}
\end{equation}
where $\tilde{\mathbf{D}}_v=\mathbf{D}_v-\hat{\mathbf{D}}_v$.

Choose the second Lyapunov function as $V_2:=|\mathbf{e}_2|+\frac{1}{2\epsilon_1}\|\tilde{\mathbf{D}}_v\|^2$. Then:
\begin{align}\label{Eq:dotV2}
  &\dot{V}_2=\mbox{sgn}^T(\mathbf{e}_2)(\mathcal{H}\mathbf{e}_3-\mathcal{H}\mathbf{K}_2\mathbf{e}_2-\mathbf{K}_1^2\mathbf{e}_1+\mathbf{K}_1\mathbf{e}_2+\mathcal{H}\tilde{\mathbf{D}}_v) \nonumber \\
    &+\frac{1}{\epsilon_1}\left(\dot{\mathbf{D}}_v-\dot{\hat{\mathbf{D}}}_v\right)^T\tilde{\mathbf{D}}_v   \nonumber     \\
    &\leq |\mathcal{H}\mathbf{e}_3|-\lambda_{\min}(\mathcal{H}\mathbf{K}_2-\mathbf{K}_1)|\mathbf{e}_2|+\lambda_{\max}(\mathbf{K}_1^2)|\mathbf{e}_1|   \nonumber     \\
    &+\left(\mbox{sgn}^T(\mathbf{e}_2)\mathcal{H}-\frac{1}{\epsilon_1}\dot{\hat{\mathbf{D}}}_v^T\right)\tilde{\mathbf{D}}_v+\frac{1}{\epsilon_1}\bar{\Delta}_v\|\tilde{\mathbf{D}}_v\|.
\end{align}

Design the adaptive law to get $\hat{\mathbf{D}}_v$ as:
\begin{align}\label{Eq:Dv_hat}
    \dot{\hat{\mathbf{D}}}_v&=-\epsilon_1\kappa_1\hat{\mathbf{D}}_v+\epsilon_1\mathcal{H}^T\mbox{sgn}(\mathbf{e}_2),   \nonumber  \\
    &=-\epsilon_1\kappa_1\hat{\mathbf{D}}_v+\epsilon_1\mathcal{H}\mbox{sgn}(\mathbf{e}_2)
\end{align}
where $\epsilon_1$, $\kappa_1>0$ are scalar control parameters.

Substitute \eqref{Eq:Dv_hat} into \eqref{Eq:dotV2} to obtain:
\begin{align*}
    &\dot{V}_2\leq |\mathcal{H}\mathbf{e}_3|-\lambda_{\min}(\mathcal{H}\mathbf{K}_2-\mathbf{K}_1)|\mathbf{e}_2|+\lambda_{\max}(\mathbf{K}_1^2)|\mathbf{e}_1|+\\
    &\kappa_1\hat{\mathbf{D}}_v^T\tilde{\mathbf{D}}_v+\frac{1}{\epsilon_1}\bar{\Delta}_v\|\tilde{\mathbf{D}}_v\|\\
    \leq& |\mathcal{H}\mathbf{e}_3|-\lambda_{\min}(\mathcal{H}\mathbf{K}_2-\mathbf{K}_1)|\mathbf{e}_2|+\lambda_{\max}(\mathbf{K}_1^2)|\mathbf{e}_1|+\\
    &\frac{1}{2}\kappa_1\|\mathbf{D}_v\|^2-\frac{1}{2}\kappa_1\|\tilde{\mathbf{D}}_v\|^2+\frac{1}{2\epsilon_1}\bar{\Delta}_v^2+\frac{1}{2\epsilon_1}\|\tilde{\mathbf{D}}_v\|^2\\
    \leq& -\min\{\lambda_{\min}(\mathcal{H}\mathbf{K}_2-\mathbf{K}_1),\epsilon_1\kappa_1-1\}\Big(|\mathbf{e}_2|+\frac{1}{2\epsilon_1}\|\tilde{\mathbf{D}}_v\|^2\Big)\\
    &+|\mathcal{H}\mathbf{e}_3|+\lambda_{\max}(\mathbf{K}_1^2)|\mathbf{e}_1|+\frac{1}{2}\kappa_1\Delta_v^2+\frac{1}{2\epsilon_1}\bar{\Delta}_v^2.
\end{align*}

\paragraph*{\textbf{Step 3}}  Recall that $\mathbf{e}_3=\mathbf{e}_a-\mathbf{e}_a^*$, where the virtual controller is as \eqref{Eq:2nd_virtual_law}. Consequently,
\begin{align*}
    &\dot{\mathbf{e}}_3=\dot{\mathbf{e}}_a-\eta\mathbf{e}_a^*+\eta\mathbf{e}_a^*-\dot{\mathbf{e}}_a^*\\
        &=-\mathbf{F}-\mathbf{G}\mathbf{u}+\mathbf{D_a}+\eta(\mathbf{K}_2\mathbf{e}_2+\hat{\mathbf{D}}_v)+(\eta\mathbf{e}_a^*-\dot{\mathbf{e}}_a^*),
\end{align*}
where $\eta>0$ is a scalar parameter for the filter; the intermediate virtual controller $\mathbf{e}_a^*$ is introduced to use the information in Step 2 so as to avoid differentiation.

Choose the third Lyapunov function as $V_3:=|\mathcal{H}\mathbf{e}_3|+\frac{1}{2\epsilon_2}\tilde{\mathbf{D}}_a^T\mathcal{H}\tilde{\mathbf{D}}_a$ (recall that $\mathcal{H}>0$), where the estimation error $\tilde{\mathbf{D}}_a=\mathbf{D}_a-\hat{\mathbf{D}}_a$, then:
\begin{align}\label{Eq:V3_dot}
    \dot{V}_3=&\mbox{sgn}^T(\mathcal{H}\mathbf{e}_3)\mathcal{H}(-\mathbf{F}-\mathbf{G}\mathbf{u}+\mathbf{D}_a+\eta(\mathbf{K}_2\mathbf{e}_2+ \nonumber      \\
    &\hat{\mathbf{D}}_v))+\frac{1}{\epsilon_2}\left(\dot{\mathbf{D}}_a-\dot{\hat{\mathbf{D}}}_a\right)^T\mathcal{H}\tilde{\mathbf{D}}_a+\Delta^*,
\end{align}
where $\Delta^*:=\max_{t\geq 0}|\mathcal{H}(\eta\mathbf{e}_a^*-\dot{\mathbf{e}}_a^*)|$ with $\Delta^*<\infty$, since $(\eta\mathbf{e}_a^*-\dot{\mathbf{e}}_a^*)\in\mathcal{L}_{\infty}$ holds \cite{liu2021distributed}.

Design the actual control law as follows:
\begin{equation}\label{Eq:actual_law}
    \mathbf{u}=\mathbf{G}^{-1}(-\mathbf{F}+\mathbf{K}_3\mathcal{H}\mathbf{e}_3+\eta\hat{\mathbf{D}}_v+\hat{\mathbf{D}}_a),
\end{equation}
where $\mathbf{K}_3>0$ is a diagonal control parameter matrix.

Then, substitute \eqref{Eq:actual_law} into \eqref{Eq:V3_dot} and obtain:
\begin{align}\label{Eq:V3_dot1}
    &\dot{V}_3\leq \mbox{sgn}^T(\mathcal{H}\mathbf{e}_3)\mathcal{H}(-\mathbf{K}_3\mathcal{H}\mathbf{e}_3+\tilde{\mathbf{D}}_a)+\eta|\mathbf{K}_2\mathbf{e}_2|-  \nonumber  \\
    &\frac{1}{\epsilon_2}\dot{\hat{\mathbf{D}}}_a^T\mathcal{H}\tilde{\mathbf{D}}_a+\frac{1}{\epsilon_2}\dot{\mathbf{D}}_a^T\mathcal{H}^{\frac{1}{2}}\mathcal{H}^{\frac{1}{2}}\tilde{\mathbf{D}}_a+\Delta^*   \nonumber      \\
    &\leq -\lambda_{\min}(\mathcal{H}\mathbf{K}_3)|\mathcal{H}\mathbf{e}_3|+\eta\lambda_{\max}(\mathbf{K}_2)|\mathbf{e}_2|+\frac{1}{2\epsilon_2}\lambda_{\max}(\mathcal{H})\bar{\Delta}_a^2+  \nonumber   \\
    &\Delta^*+\frac{1}{2\epsilon_2}\tilde{\mathbf{D}}_a^T\mathcal{H}\tilde{\mathbf{D}}_a+\left(\mbox{sgn}^T(\mathcal{H}\mathbf{e}_3)-\frac{1}{\epsilon_2}\dot{\hat{\mathbf{D}}}_a^T\right)\mathcal{H}\tilde{\mathbf{D}}_a.
\end{align}

Design the adaptive law to get $\hat{\mathbf{D}}_a$ as:
\begin{equation}\label{Eq:Da_hat}
    \dot{\hat{\mathbf{D}}}_a=-\epsilon_2\kappa_2\hat{\mathbf{D}}_a+\epsilon_2\mbox{sgn}(\mathcal{H}\mathbf{e}_3),
\end{equation}
where $\epsilon_2$, $\kappa_2>0$ are scalar control parameters.

Substitute \eqref{Eq:Da_hat} into \eqref{Eq:V3_dot1} to obtain:
\begin{align*}
    &\dot{V}_3\leq -\lambda_{\min}(\mathcal{H}\mathbf{K}_3)|\mathcal{H}\mathbf{e}_3|+\eta\lambda_{\max}(\mathbf{K}_2)|\mathbf{e}_2|+\\
    &\frac{1}{2\epsilon_2}\lambda_{\max}(\mathcal{H})\bar{\Delta}_a^2+\frac{1}{2\epsilon_2}\tilde{\mathbf{D}}_a^T\mathcal{H}\tilde{\mathbf{D}}_a+\kappa_2\hat{\mathbf{D}}_a^T\mathcal{H}^{\frac{1}{2}}\mathcal{H}^{\frac{1}{2}}\tilde{\mathbf{D}}_a+\Delta^*\\
    &\leq -\lambda_{\min}(\mathcal{H}\mathbf{K}_3)|\mathcal{H}\mathbf{e}_3|+\eta\lambda_{\max}(\mathbf{K}_2)|\mathbf{e}_2|+\frac{1}{2\epsilon_2}\lambda_{\max}(\mathcal{H})\bar{\Delta}_a^2\\
    &+\frac{1}{2\epsilon_2}\tilde{\mathbf{D}}_a^T\mathcal{H}\tilde{\mathbf{D}}_a+\frac{1}{2}\kappa_2\lambda_{\max}(\mathcal{H})\|\mathbf{D}_a\|^2-\frac{1}{2}\kappa_2\tilde{\mathbf{D}}_a^T\mathcal{H}\tilde{\mathbf{D}}_a+\Delta^*\\
    &\leq -\min\{\lambda_{\min}(\mathcal{H}\mathbf{K}_3),\epsilon_2\kappa_2-1\}\Big(|\mathcal{H}\mathbf{e}_3|+\frac{1}{2\epsilon_2}\tilde{\mathbf{D}}_a^T\mathcal{H}\tilde{\mathbf{D}}_a\Big)+\\
    &\eta\lambda_{\max}(\mathbf{K}_2)|\mathbf{e}_2|+\frac{1}{2}\kappa_2\lambda_{\max}(\mathcal{H})\Delta_a^2+\frac{1}{2\epsilon_2}\lambda_{\max}(\mathcal{H})\bar{\Delta}_a^2+\Delta^*.
\end{align*}

\begin{remark}
    In the above controller design, the actual control law \eqref{Eq:actual_law} and the adaptive laws \eqref{Eq:Dv_hat} and \eqref{Eq:Da_hat} involve terms that are linear in $\mathcal{H}$. This property promotes the distributed nature of our controller. Note that the $i$\textsuperscript{th} component of \eqref{Eq:Dv_hat}, \eqref{Eq:Da_hat}, \eqref{Eq:actual_law} for any $i\in\mathcal{I}_N$ are:
    \begin{align*}
    \dot{\hat{D}}_{vi}=&-\epsilon_1\kappa_1\hat{D}_{vi}+\frac{\epsilon_1}{|\mathbf{e}_2|}\Big(P_ie_{2i}+\sum_{j\in\mathcal{I}_N\backslash\{i\}}A_{ij}(e_{2i}-e_{2j})\Big),    \\
    \dot{\hat{D}}_{ai}=&-\epsilon_2\kappa_2\hat{D}_{ai}+\frac{\epsilon_2}{|\mathcal{H}\mathbf{e}_3|}\Big(P_ie_{3i}+\sum_{j\in\mathcal{I}_N\backslash\{i\}}A_{ij}(e_{3i}-e_{3j})\Big),   \\
    u_i=&g_i^{-1}\Big(-f_i+K_{3i}(P_ie_{3i}+\sum_{j\in\mathcal{I}_N\backslash\{i\}}A_{ij}(e_{3i}-e_{3j}))+    \\
    &\eta\hat{D}_{vi}+\hat{D}_{ai}\Big),
\end{align*}
respectively. From these expressions, it is clear that the actual control law \eqref{Eq:actual_law} and the adaptive laws \eqref{Eq:Dv_hat} and \eqref{Eq:Da_hat} can in fact be implemented in a distributed manner when each agent (follower vehicle) is provided with the information $|\mathbf{e}_2|$ and $|\mathcal{H}\mathbf{e}_3|$ (e.g., via some centralized external entities or some decentrally executed distributed consensus schemes).
\end{remark}

\begin{remark}
    Since our controller is designed based on vector Lyapunov function-based analysis, the controller for each vehicle can be stacked into a vector form. As are observed in the actual control law \eqref{Eq:actual_law} and the adaptive laws \eqref{Eq:Dv_hat} and \eqref{Eq:Da_hat}, these equations are still applicable as long as the topology matrix $\mathcal{H}$ is given without redesigning the controllers for the entire platoon. This indicates the compositionality of our controller and highlights the advantage of using our controller compared with traditional scalar Lyapunov function-based backstepping designs \cite{zhu2018distributed,zhu2020v2v,liu2021distributed}.

    
\end{remark}

\subsection{Stability Analysis Under VSLF}

\begin{theorem}\label{theorem:controller}
    For the VSLF of the form $V:=[V_1,V_2,V_3]^T$, where $V_1 := |\mathbf{e}_1|$, $V_2 := |\mathbf{e}_2|+\frac{1}{2\epsilon_1}\|\tilde{\mathbf{D}}_v\|^2$ and $V_3:=|\mathcal{H}\mathbf{e}_3|+\frac{1}{2\epsilon_2}\tilde{\mathbf{D}}_a^T\mathcal{H}\tilde{\mathbf{D}}_a$, the internal stability and $l_2$ weak string stability of the platoon error dynamics \eqref{Eq:platoon_error_vector} can be guaranteed under the distributed adaptive backstepping controller \eqref{Eq:actual_law} with virtual controllers \eqref{Eq:1st_virtual_law} and \eqref{Eq:2nd_virtual_law} and adaptive laws \eqref{Eq:Dv_hat} and \eqref{Eq:Da_hat}.
\end{theorem}
\begin{proof}
    From the proposed backstepping controller design, we can summarize the following relationships:
\begin{align}
    \dot{V}_1\leq&-\lambda_{\min}(\mathbf{K}_1)V_1+V_2,    \nonumber\\
        \dot{V}_2\leq& -\min\{\lambda_{\min}(\mathcal{H}\mathbf{K}_2-\mathbf{K}_1),\epsilon_1\kappa_1-1\}V_2+V_3       \nonumber            \\        &+\lambda_{\max}(\mathbf{K}_1^2)V_1+\frac{1}{2}\kappa_1\Delta_v^2+\frac{1}{2\epsilon_1}\bar{\Delta}_v^2,   \nonumber \\
        \dot{V}_3\leq& -\min\{\lambda_{\min}(\mathcal{H}\mathbf{K}_3),\epsilon_2\kappa_2-1\}V_3+\eta\lambda_{\max}(\mathbf{K}_2)V_2   \nonumber\\
        &+\frac{1}{2}\kappa_2\lambda_{\max}(\mathcal{H})\Delta_a^2+\frac{1}{2\epsilon_2}\lambda_{\max}(\mathcal{H})\bar{\Delta}_a^2+\Delta^*.
\end{align}

Thus, we can select the VSLF as $V:=[V_1,V_2,V_3]^T$. Then,
\begin{equation}
    \begin{split}
        \dot{V}\preceq&\ \Gamma V+\mathbf{b},
    \end{split}
\end{equation}
where $\Gamma\in\mathbb{R}^{3\times 3}$ and $\mathbf{b}\in\mathbb{R}_+^3$ are defined respectively as:
\begin{gather*}
    \Gamma:= \scriptsize\left[\begin{array}{ccc}
            -\lambda_{\min}(\mathbf{K}_1) & 1 & 0 \\
            \lambda_{\max}(\mathbf{K}_1^2) & -\lambda_{\min}(\mathcal{H}\mathbf{K}_2-\mathbf{K}_1) & 1 \\
            0 & \eta\lambda_{\max}(\mathbf{K}_2) & -\lambda_{\min}(\mathcal{H}\mathbf{K}_3)
        \end{array}\right],       \\
    \mathbf{b}:= \scriptsize\left[b_1,b_2,b_3\right]^T,
\end{gather*}
where $\Gamma$ is selected as a Hurwitz matrix; $b_1:=0$, $b_2:= \frac{1}{2}\kappa_1\Delta_v^2+\frac{1}{2\epsilon_1}\bar{\Delta}_v^2$, $b_3:=\frac{1}{2}\kappa_2\lambda_{\max}(\mathcal{H})\Delta_a^2+\frac{1}{2\epsilon_2}\lambda_{\max}(\mathcal{H})\bar{\Delta}_a^2+\Delta^*$, under the conditions:
\begin{gather}
    \lambda_{\min}(\mathcal{H}\mathbf{K}_2-\mathbf{K}_1)\leq\epsilon_1\kappa_1-1, \\
    \lambda_{\min}(\mathcal{H}\mathbf{K}_3)\leq\epsilon_2\kappa_2-1.
\end{gather}

Besides, we can choose a comparison system as:
\begin{equation}\label{Eq:comparison_system_z}
    \dot{\mathbf{z}}=\Gamma\mathbf{z}+\mathbf{b}.
\end{equation}

Since the matrix $\Gamma$ is Hurwitz, we know that the system \eqref{Eq:comparison_system_z} is ISS from the input $\mathbf{b}$ to the state $\mathbf{z}$. Moreover, if the initial condition satisfies $\mathbf{z}(0)=V(0)$, we have $V\preceq \mathbf{z}$, for all $t\geq 0$. So the errors $\mathbf{e}_1$, $\mathbf{e}_2$ and $\mathbf{e}_3$ are all bounded, i.e., $\mathbf{e}_1$, $\mathbf{e}_2$, $\mathbf{e}_3\in\mathcal{L}_{\infty}$, for all $t\geq 0$. Since $\mathbf{z}$ is ISS from $\mathbf{b}$ to $\mathbf{z}$, it implies that the closed loop system under our designed controller is ISS from the input $(\Delta_v,\bar{\Delta}_v,\Delta_a,\bar{\Delta}_a,\Delta^*)$ to the states $\mathbf{e}_1$, $\mathbf{e}_2$ and $\mathbf{e}_3$. This implies the internal stability of the platoon. Furthermore, the conditions of our proposed VSLF notion in Def. \ref{def:VSLF} are satisfied. Thus, the $l_2$ weak string stability is satisfied. 
\end{proof}

\begin{remark}
    It may seem that the three Lyapunov functions that we chose in the controller design violate the condition \eqref{def:state_error_bound} of the Def. \ref{def:VSLF}, since we may not find a such class-$\mathcal{K}_{\infty}$ function $\overline{\alpha}(\cdot)$. This is because our definition of VSLF is provided in a general robust form without considering the estimation errors for the disturbances. If we just design a robust platoon controller without involving the adaptive laws \eqref{Eq:Dv_hat} and \eqref{Eq:Da_hat}, then we can simply select those three Lyapunov functions as $V_1=|\mathbf{e}_1|$, $V_2=|\mathbf{e}_2|$ and $V_3=|\mathcal{H}\mathbf{e}_3|$, respectively. Therefore, it is readily seen that they satisfy the condition \eqref{def:state_error_bound} of the Def. \ref{def:VSLF}. For the case that we involve the adaptive laws \eqref{Eq:Dv_hat} and \eqref{Eq:Da_hat}, we can simply modify the condition \eqref{def:state_error_bound} by finding an upper bound as $\overline{\alpha}_1(|\mathbf{e}|)+\overline{\alpha}_2(\|\tilde{\mathbf{D}}\|)$, where $\tilde{\mathbf{D}}$ is the lumped estimation errors. It can be easily seen from the proof of Theorem \ref{theorem:VSS_DSS} that this modified form of VSLF also implies $l_2$ weak string stability with a different comparison function in \eqref{Eq:ISSS_condition}.
\end{remark}

\section{Simulation Examples}\label{sec:simulation}

In this section, we provide simulation examples for the control of a platoon to verify the effectiveness of our proposed controller. We consider a platoon with four homogeneous vehicles following a leader and the dynamics of each vehicle in the platoon is of the form \eqref{Eq:platoon_longitudinal} with parameters as $m_i=1500\ \mbox{kg}$, $\tau_i=0.25\ \mbox{s}$, $A_{f,i}=2.2\ \mbox{m}^2$, $\rho=0.78\ \mbox{kg/m}^3$, $C_{d,i}=0.35$ and $C_{r,i}=0.067$.

The dynamics of the leader is of the form \eqref{Eq:leader_dynamics} and it is assumed to generate a reference velocity profile $v_0$ as:
\begin{equation}
    v_0: \left \{\begin{array}{ll}
        15\ \mbox{m/s}, & 0\leq t\leq 5\ \mbox{s}\\
        (2t+5)\ \mbox{m/s}, & 5\leq t\leq 10\ \mbox{s}\\
        25\ \mbox{m/s}, & 10\leq t\leq 15\ \mbox{s}\\
        (-t+40)\ \mbox{m/s}, & 15\leq t\leq 20\ \mbox{s}\\
        20\ \mbox{m/s}. & 20\leq t\leq 30\ \mbox{s}
        \end{array}\right.
\end{equation}
The initial position of the leading vehicle is assumed to be $x_0(0)=20\ \mbox{m}$, the initial conditions of the following vehicles are assumed to be $\mathbf{x}(0)=[15,\ 10,\ 5,\ 0]^T \ \mbox{m}$, $\mathbf{v}(0)=\mathbf{a}(0)=[0,\ 0,\ 0,\ 0]^T$.

To illustrate the effectiveness of the proposed control method in this paper, we compare our simulation results with those of the work in \cite{wang2020finitetime} which also considers the mismatched disturbances in the vehicle dynamics and uses a compositional and distributed controller. In particular, we compare two types of information flow topologies for the platoon as shown in Fig. \ref{fig:topologies_in_simulation}, i.e., bidirectional-leader type (denoted as bl-type) and bidirectional type (denoted as b-type) \cite{feng2019string}. Thus, the Laplacian matrix and the pinning matrix are respectively as follows:
\begin{gather}
    \mathcal{L}_{bl}=\begin{bmatrix}
    1 & -1 & 0 & 0 \\
    -1 & 2 & -1 & 0 \\
    0 & -1 & 2 & -1 \\
    0 & 0 & -1 & 1
    \end{bmatrix},\ \ 
    \mathcal{P}_{bl}=\mathbf{I}_{4\times 4}.    \\
    \mathcal{L}_b=\begin{bmatrix}
    1 & -1 & 0 & 0 \\
    -1 & 2 & -1 & 0 \\
    0 & -1 & 2 & -1 \\
    0 & 0 & -1 & 1
    \end{bmatrix},\ \
    \mathcal{P}_b=\mbox{diag}\big\{[1,\ 0,\ 0,\ 0]\big\}.
\end{gather}

\begin{figure}[!t]
    \centering
    \includegraphics[width=0.9\linewidth]{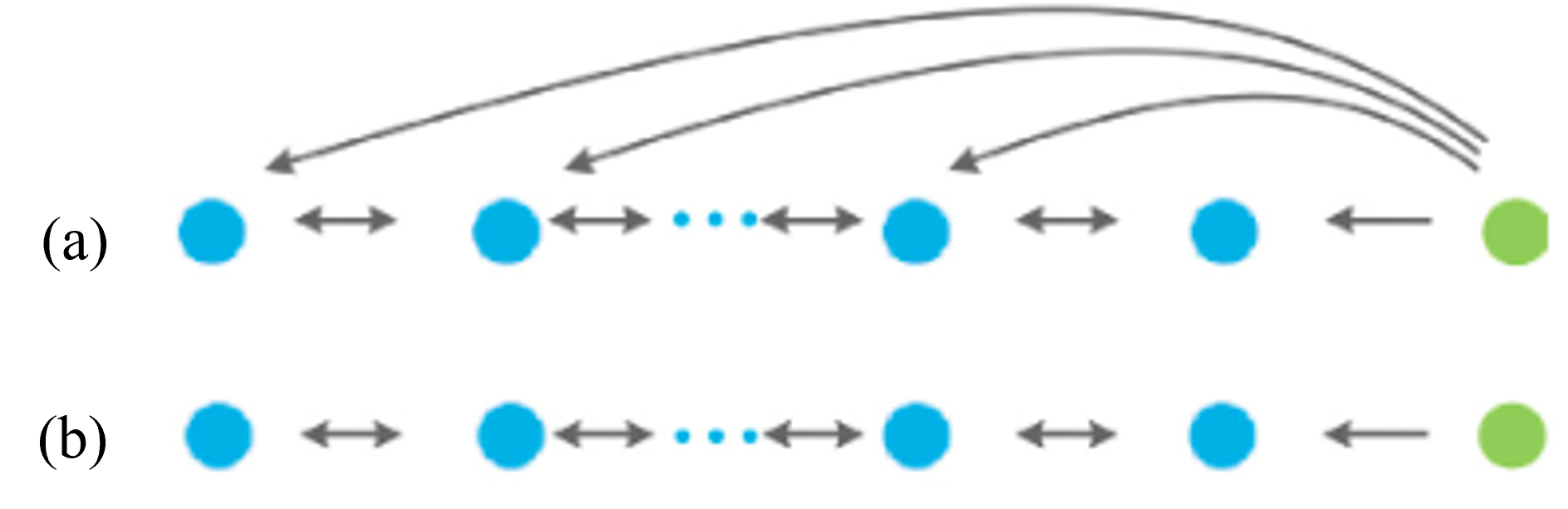}
    \caption{ (a) Bidirectional-leader type (bl-type) information flow topology; (b) bidirectional type (b-type) information flow topology.}
    \label{fig:topologies_in_simulation}
\end{figure}

Besides, the parameters of the proposed controller \eqref{Eq:1st_virtual_law}, \eqref{Eq:2nd_virtual_law} and \eqref{Eq:actual_law} are selected as  $\mathbf{K}_1=1.5\mathbf{I}_{4\times 4}$, $\mathbf{K}_2=10\mathbf{I}_{4\times 4}$, $\mathbf{K}_3=50\mathbf{I}_{4\times 4}$, $\epsilon_1=10$, $\epsilon_2=22$, $\kappa_1=\kappa_2=0.5$, $\eta=2$ for the bl-type information flow topology and $\mathbf{K}_1$, $\mathbf{K}_2$, $\mathbf{K}_3$ and $\eta$ are selected as $\mathbf{K}_1=0.6\mathbf{I}_{4\times 4}$, $\mathbf{K}_2=25\mathbf{I}_{4\times 4}$, $\mathbf{K}_3=55\mathbf{I}_{4\times 4}$, $\eta=0.05$ for the b-type information flow topology, respectively.

The initial conditions for the adaptive laws \eqref{Eq:Dv_hat} and \eqref{Eq:Da_hat} are set as $\hat{\mathbf{D}}_v=\hat{\mathbf{D}}_a=[0,\ 0,\ 0,\ 0]^T$. The desired separation between two vehicles and the length of a vehicle are assumed to be $\delta_{di}=3\ \mbox{m}$ and $L_i=2.5\ \mbox{m}$, respectively. Moreover, we compare two types of disturbances, i.e., periodic type and aperiodic type, which are assumed to be  $\mathbf{D}_v=0.3\sin(t)\mathbf{1}_{4\times 1}$, $\mathbf{D}_a=0.2\sin(t)\mathbf{1}_{4\times 1}$ (periodic type) and $\mathbf{D}_v=\mbox{randn}(4,1)$, $\mathbf{D}_a=\mbox{randn}(4,1)$ (aperiodic type), respectively. Here, $\mbox{randn}(4,1)$ represents a Gaussian random noise vector with standard normal distribution.

Under the above parameters, the observed simulation results are shown in Fig. \ref{fig:position_position_errors_sin_b_VSS}-Fig. \ref{fig:velocity_velocity_errors_aper_s_paperB}.
According to these simulation results, it is clear that both control methods can successfully achieve position and velocity tracking. Besides, under b-type information flow topology, the performance of both methods deteriorate in terms of position and velocity tracking compared with the bl-type information flow topology (see from the tracking error plots), since there are less links between the followers and the leader.

We have quantified the overall tracking performance using the standard Root Mean Square (RMS) metric (of the tracking errors) in Table. \ref{tab:comparison_RMS}. From the reported results, it is clear that our method outperforms the method proposed in \cite{wang2020finitetime} in terms of both position and velocity tracking under bl-type information flow topology with sinusoidal disturbances or Gaussian random noise. Under b-type information flow topology, the method proposed in \cite{wang2020finitetime} has smaller position tracking errors while our proposed method has smaller velocity tracking errors. However, for this case, it is obvious from Fig. \ref{fig:position_position_errors_sin_b_VSS}-\ref{fig:velocity_velocity_errors_aper_s_VSS} that the transient and the steady-states behaviors under our proposed method are better than those under the method proposed in \cite{wang2020finitetime}, which are shown in Fig. \ref{fig:position_position_errors_sin_b_paperB}- \ref{fig:velocity_velocity_errors_aper_s_paperB}. Besides, the position tracking errors for our proposed method under b-type information flow topology are indeed comparable with those for the method proposed in \cite{wang2020finitetime}. Therefore, based on the above experiments and comparisons, we can conclude that our proposed method performs better than the method proposed in \cite{wang2020finitetime}, especially with Gaussian random noise.


\begin{figure}[!t]
    \centering
    \includegraphics[width=0.8\linewidth]{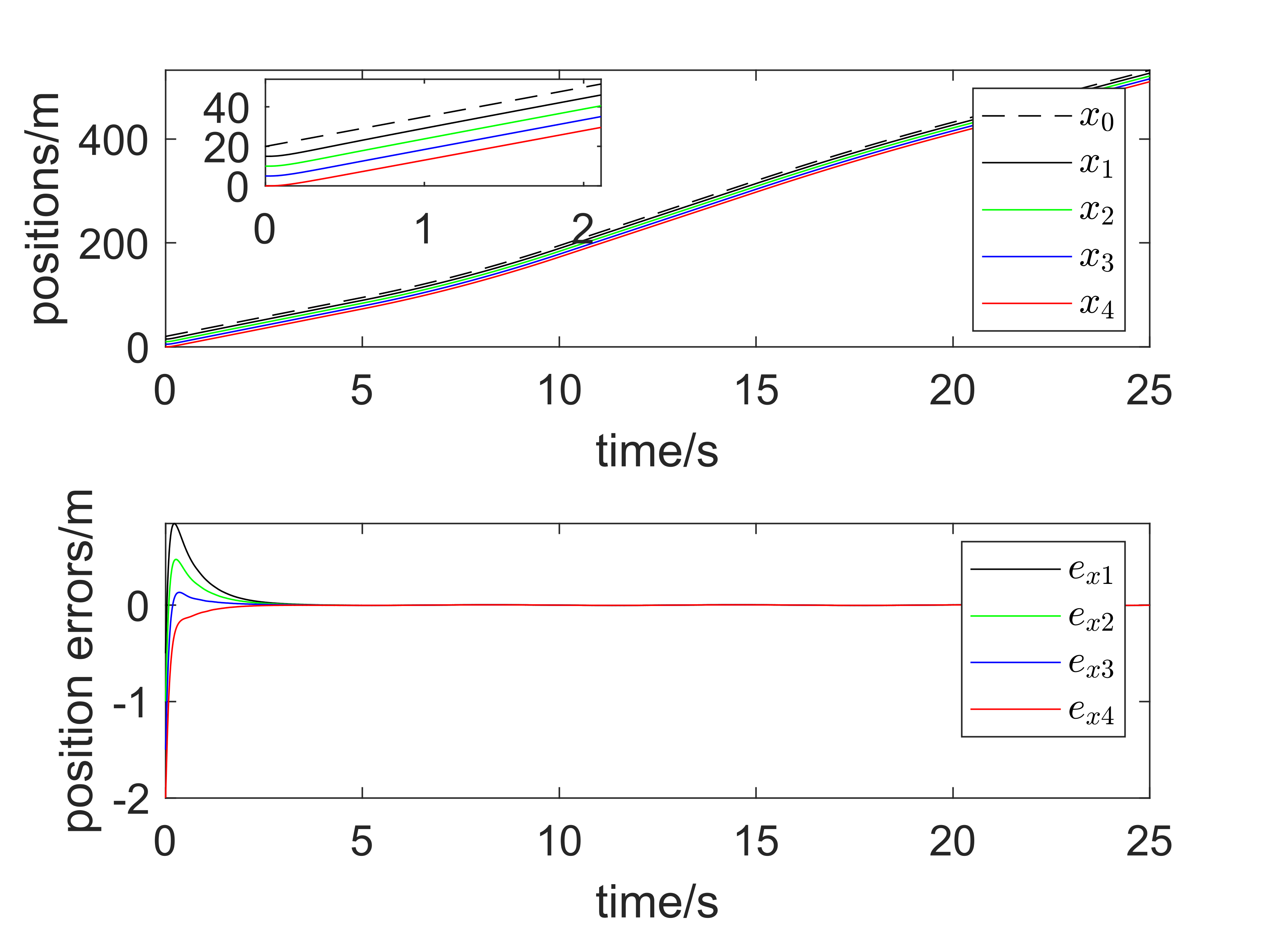}
    \caption{Results observed with our proposed controller under sinusoidal disturbances and  bidirectional-leader type information flow topology: (Top) Vehicle positions and (Bottom) Position tracking errors.}
    \label{fig:position_position_errors_sin_b_VSS}
\end{figure}
\begin{figure}[!t]
    \centering
    \includegraphics[width=0.8\linewidth]{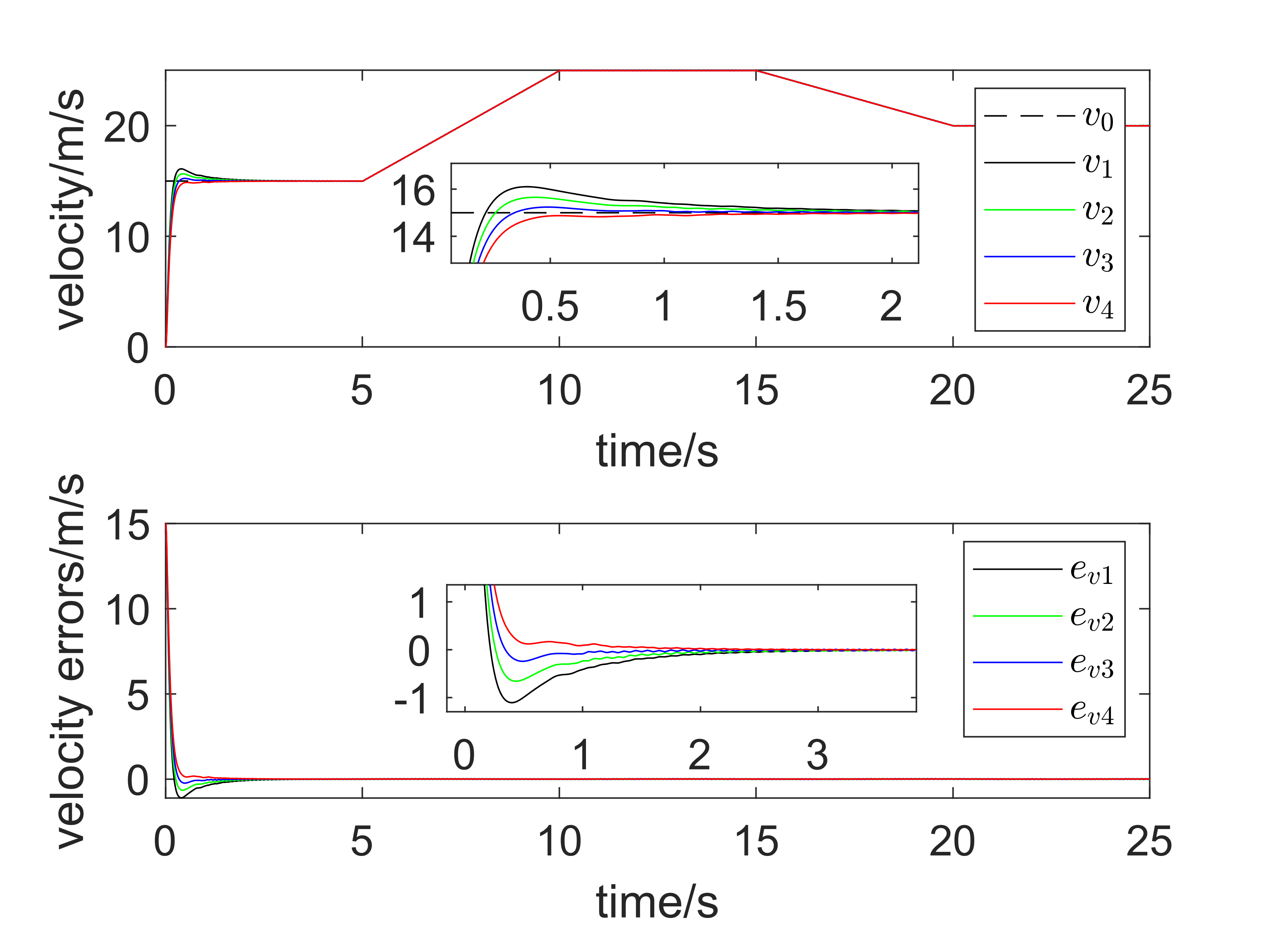}
    \caption{Results observed with our proposed controller under sinusoidal disturbances and  bidirectional-leader type information flow topology: (Top) Vehicle velocities and (Bottom) Velocity tracking errors.}
    \label{fig:velocity_velocity_errors_sin_b_VSS}
\end{figure}

\begin{figure}[!t]
    \centering
    \includegraphics[width=0.8\linewidth]{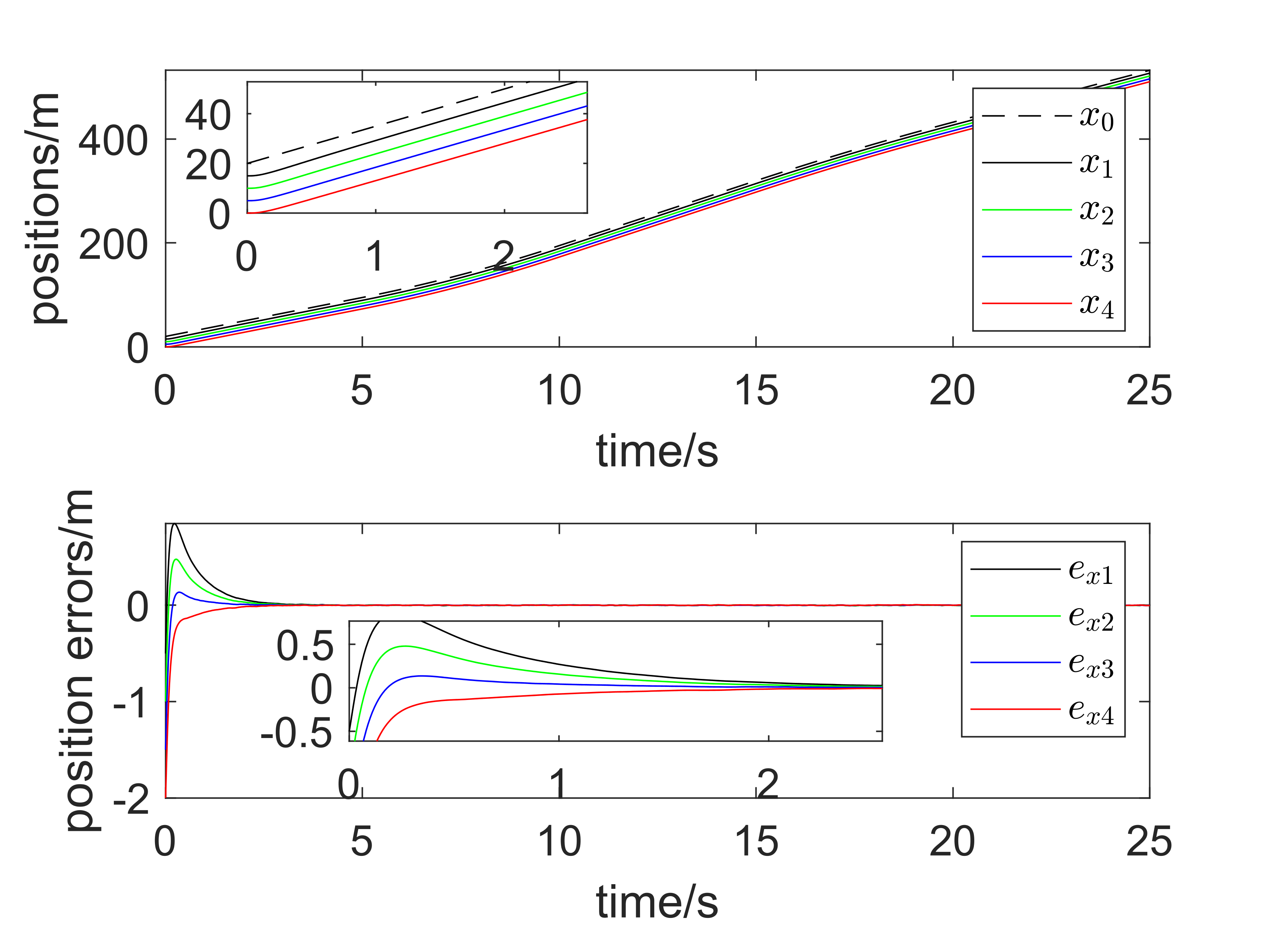}
    \caption{Results observed with our proposed controller under Gaussian random noise and bidirectional-leader type information flow topology: (Top) Vehicle positions and (Bottom) Position tracking errors.}
    \label{fig:position_position_errors_aper_b_VSS}
\end{figure}
\begin{figure}[!t]
    \centering
    \includegraphics[width=0.8\linewidth]{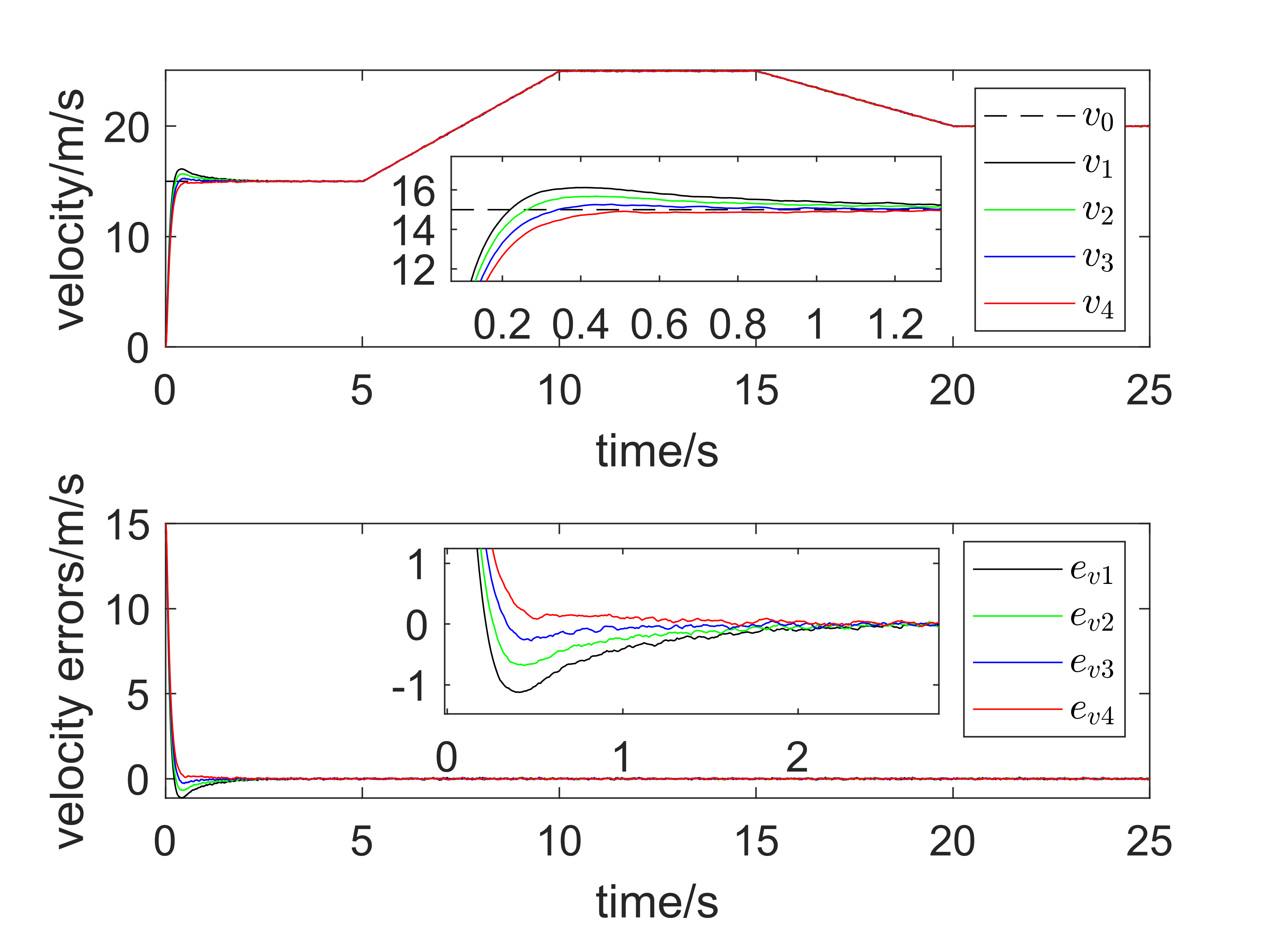}
    \caption{Results observed with our proposed controller under Gaussian random noise and bidirectional-leader type information flow topology: (Top) Vehicle velocities and (Bottom) Velocity tracking errors.}
    \label{fig:velocity_velocity_errors_aper_b_VSS}
\end{figure}

\begin{figure}[!t]
    \centering
    \includegraphics[width=0.8\linewidth]{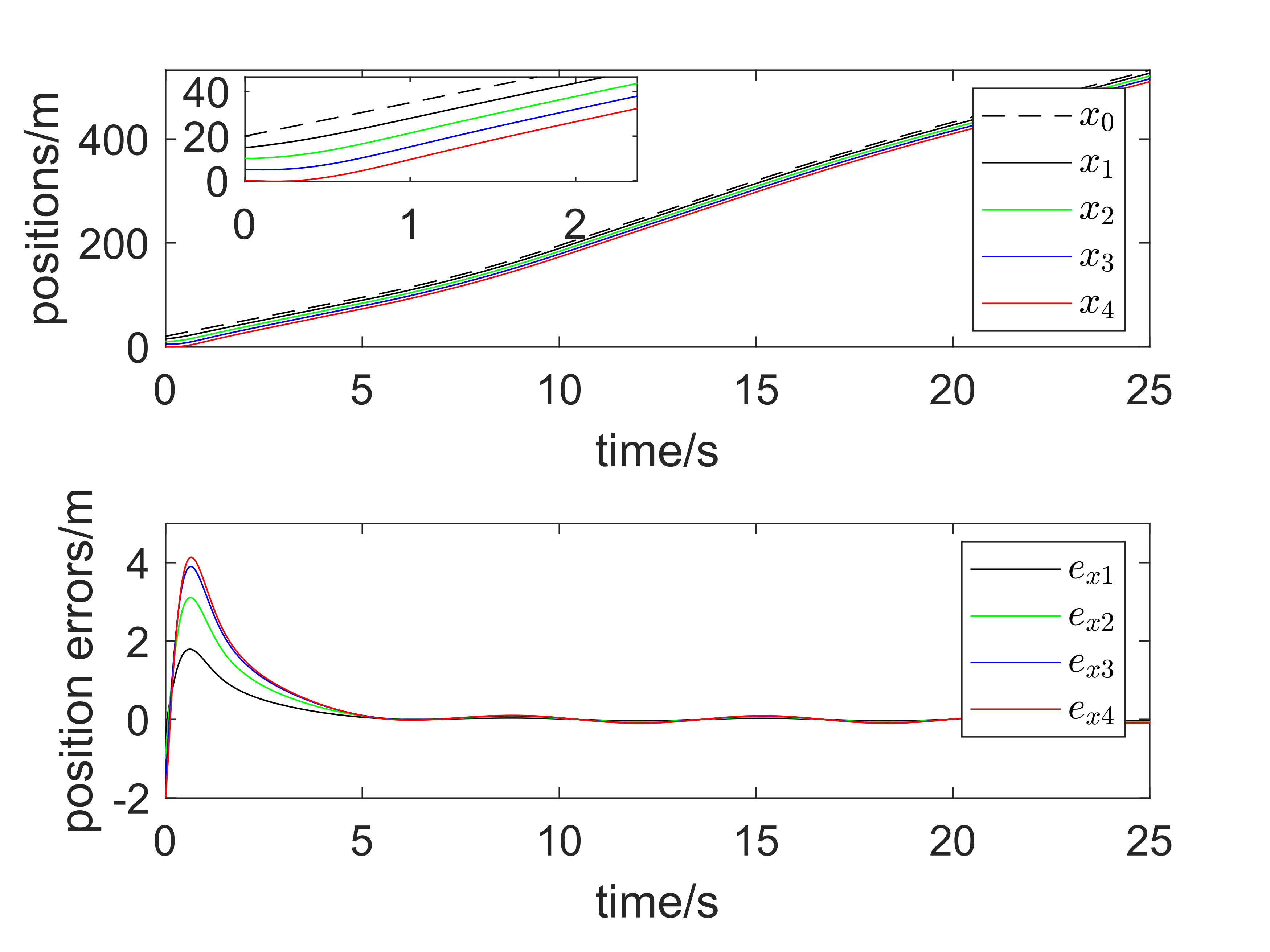}
    \caption{Results observed with our proposed controller under sinusoidal disturbances and bidirectional type information flow topology: (Top) Vehicle positions and (Bottom) Position tracking errors.}
    \label{fig:position_position_errors_sin_s_VSS}
\end{figure}
\begin{figure}[!t]
    \centering
    \includegraphics[width=0.8\linewidth]{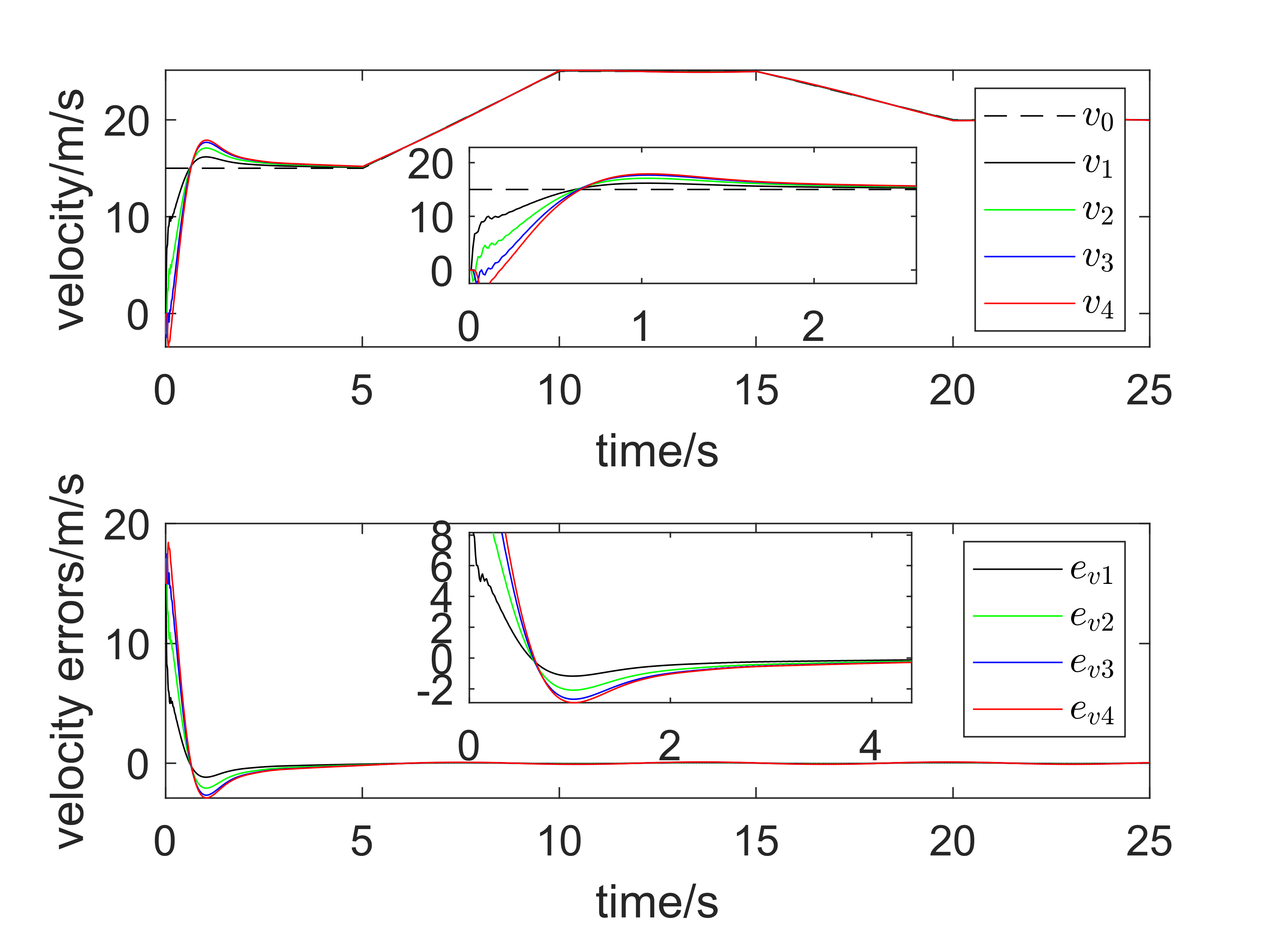}
    \caption{Results observed with our proposed controller under sinusoidal disturbances and bidirectional type information flow topology: (Top) Vehicle velocities and (Bottom) Velocity tracking errors.}
    \label{fig:velocity_velocity_errors_sin_s_VSS}
\end{figure}

\begin{figure}[!t]
    \centering
    \includegraphics[width=0.8\linewidth]{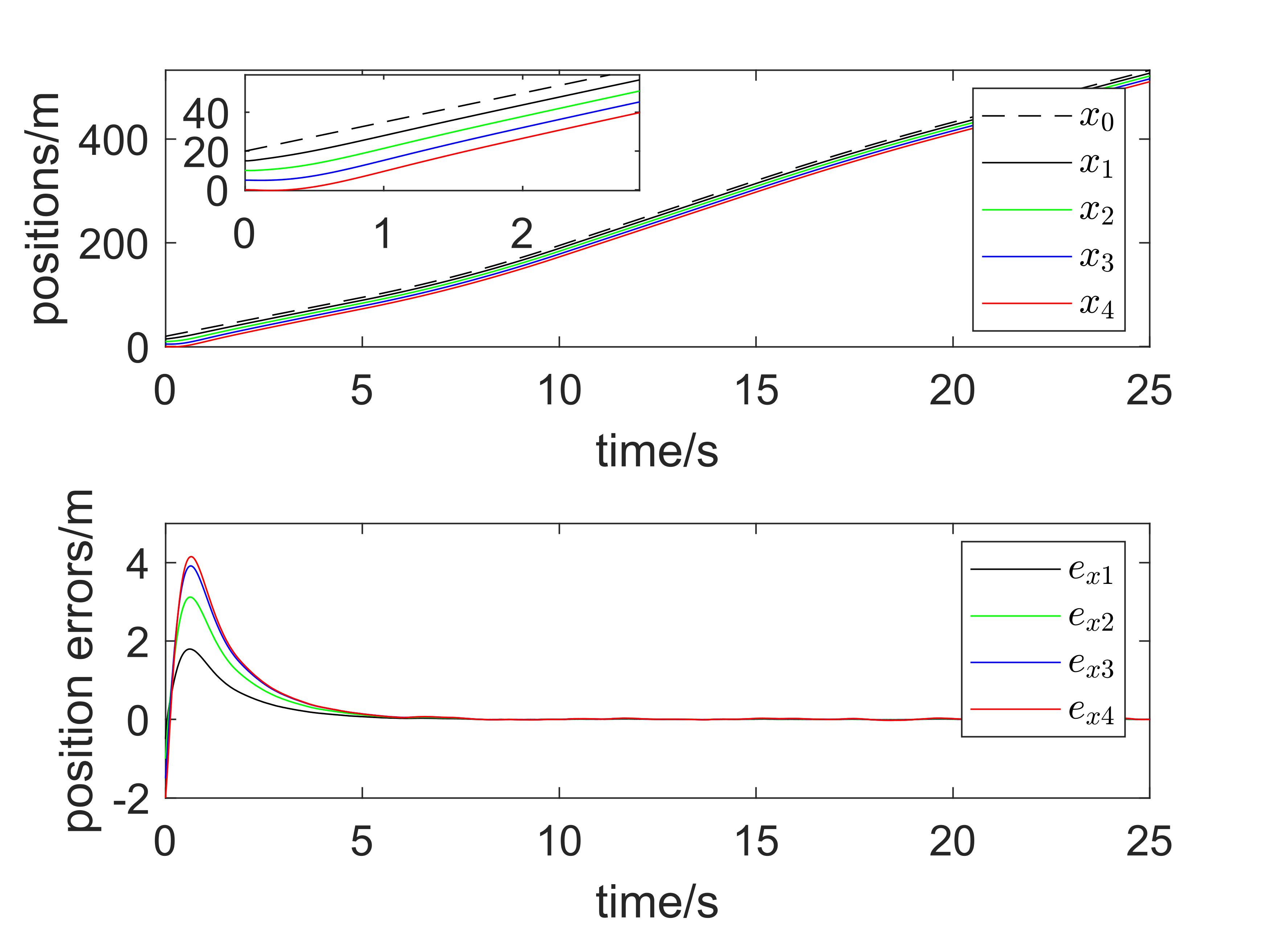}
    \caption{Results observed with our proposed controller under Gaussian random noise and bidirectional type information flow topology: (Top) Vehicle positions and (Bottom) Position tracking errors.}
    \label{fig:position_position_errors_aper_s_VSS}
\end{figure}
\begin{figure}[!t]
    \centering
    \includegraphics[width=0.8\linewidth]{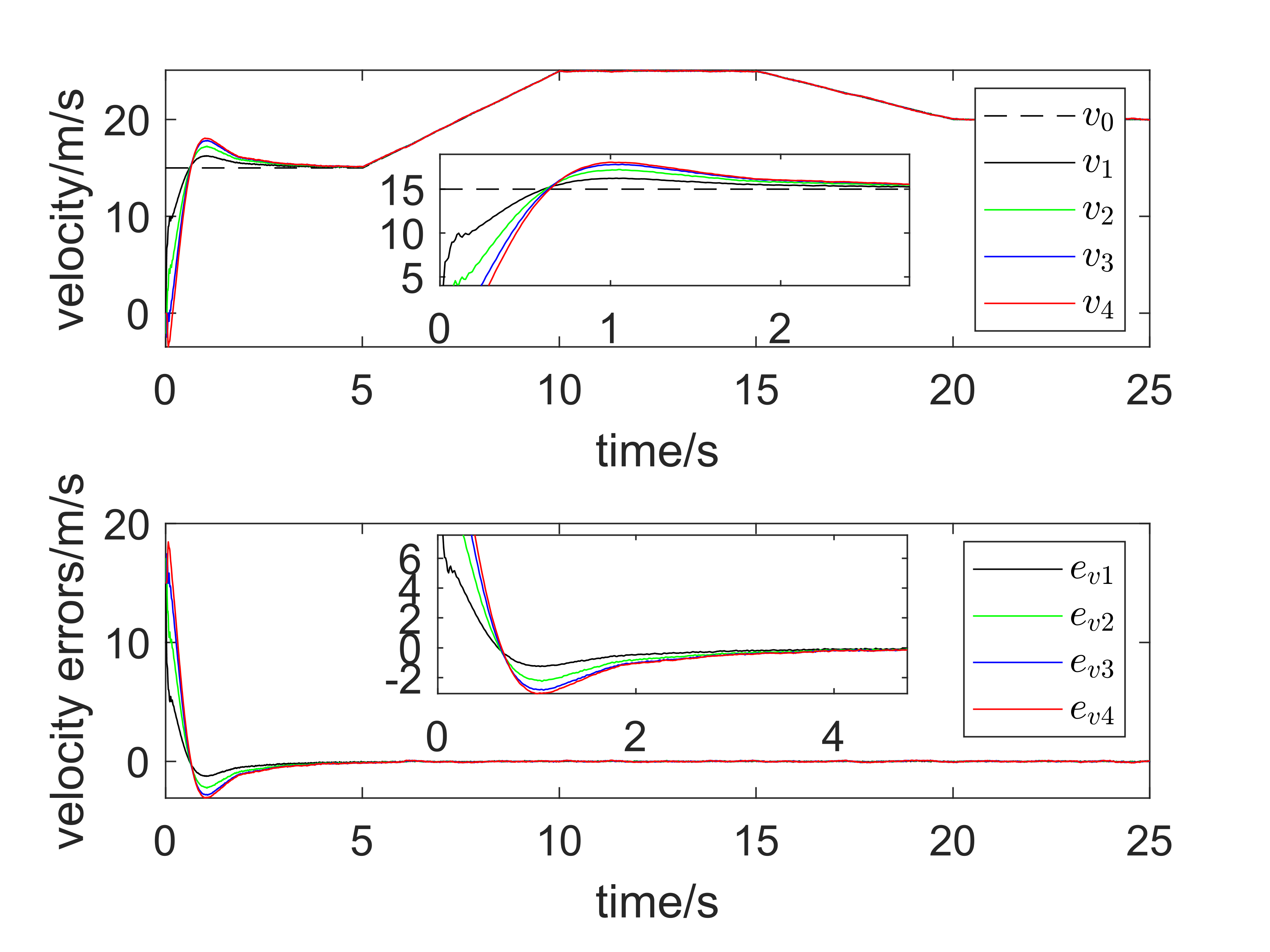}
    \caption{Results observed with our proposed controller under Gaussian random noise and bidirectional type information flow topology: (Top) Vehicle velocities and (Bottom) Velocity tracking errors.}
    \label{fig:velocity_velocity_errors_aper_s_VSS}
\end{figure}

\begin{figure}[!t]
    \centering
    \includegraphics[width=0.8\linewidth]{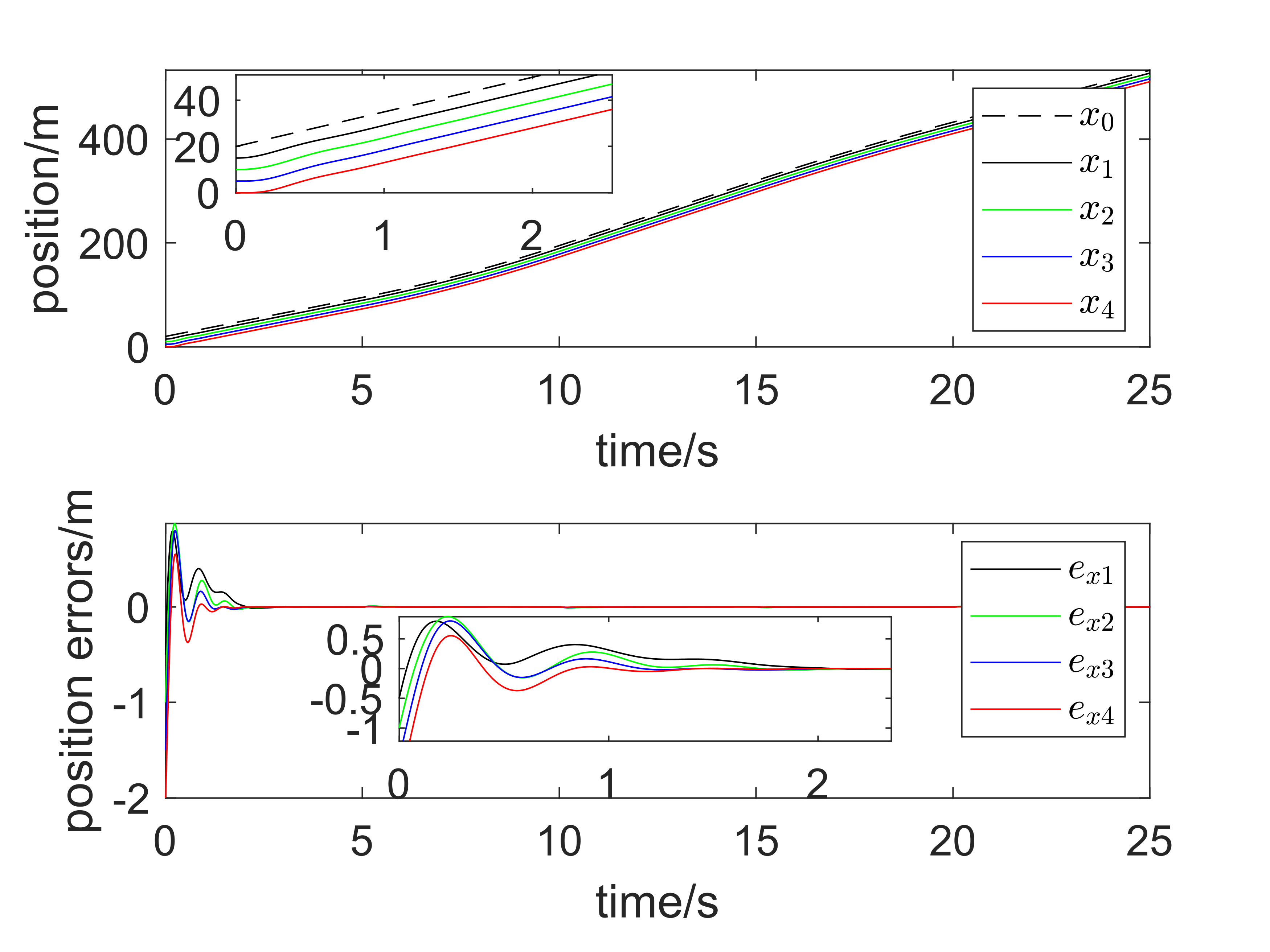}
    \caption{Results observed with the controller proposed in \cite{wang2020finitetime} under sinusoidal disturbances and bidirectional-leader type information flow topology: (Top) Vehicle positions and (Bottom) Position tracking errors.}
    \label{fig:position_position_errors_sin_b_paperB}
\end{figure}
\begin{figure}[!t]
    \centering
    \includegraphics[width=0.8\linewidth]{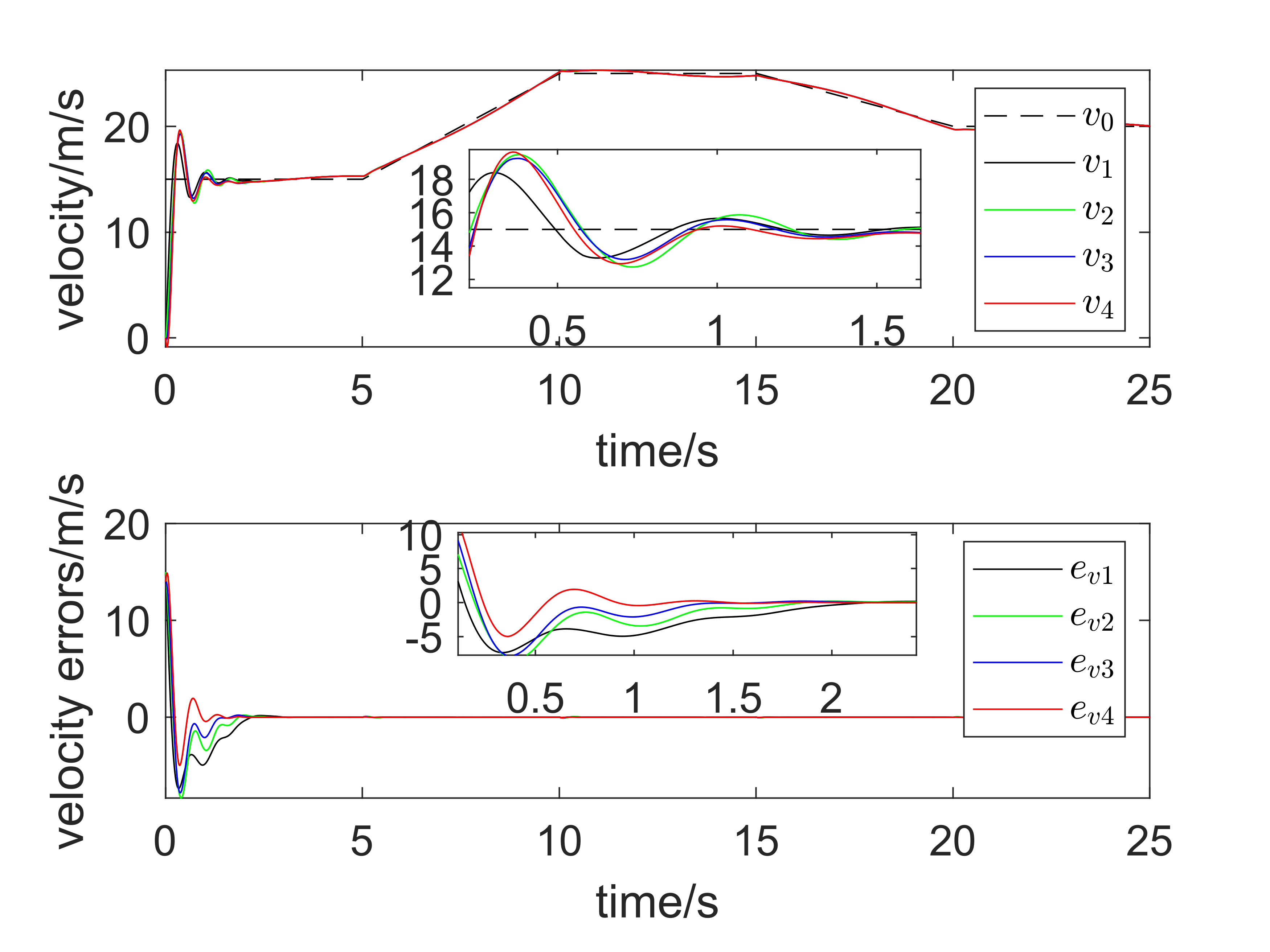}
    \caption{Results observed with the controller proposed in \cite{wang2020finitetime} under sinusoidal disturbances and bidirectional-leader type information flow topology: (Top) Vehicle velocities and (Bottom) Velocity tracking errors.}
    \label{fig:velocity_velocity_errors_sin_b_paperB}
\end{figure}

\begin{figure}[!t]
    \centering
    \includegraphics[width=0.8\linewidth]{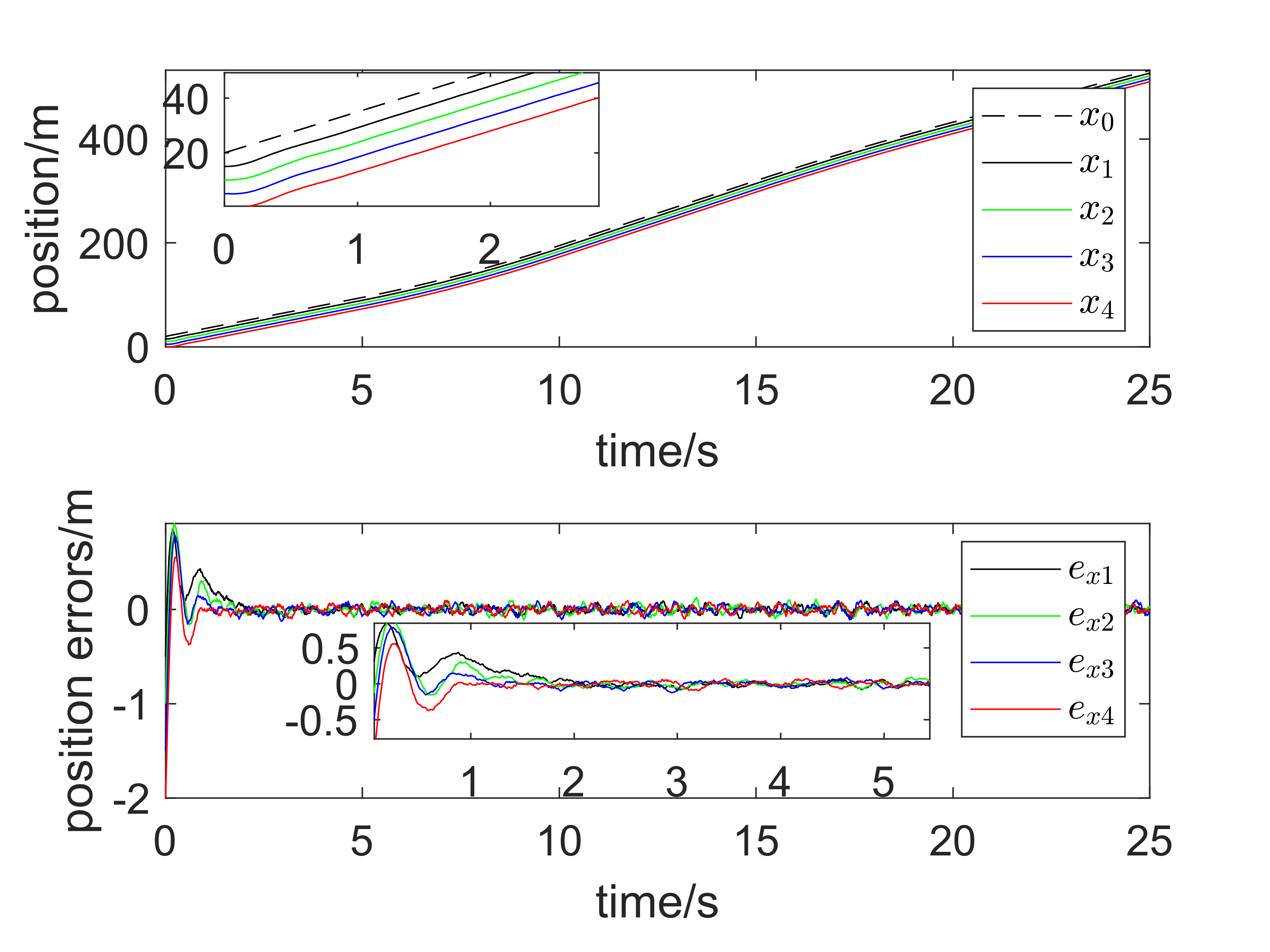}
    \caption{Results observed with the controller proposed in \cite{wang2020finitetime} under Gaussian random noise and bidirectional-leader type information flow topology: (Top) Vehicle positions and (Bottom) Position tracking errors.}
    \label{fig:position_position_errors_aper_b_paperB}
\end{figure}
\begin{figure}[!t]
    \centering
    \includegraphics[width=0.8\linewidth]{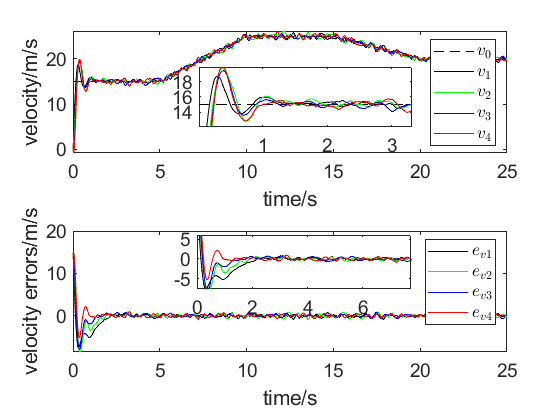}
    \caption{Results observed with the controller proposed in \cite{wang2020finitetime} under Gaussian random noise and bidirectional-leader type information flow topology: (Top) Vehicle velocities and (Bottom) Velocity tracking errors.}
    \label{fig:velocity_velocity_errors_aper_b_paperB}
\end{figure}

\begin{figure}[!t]
    \centering
    \includegraphics[width=0.8\linewidth]{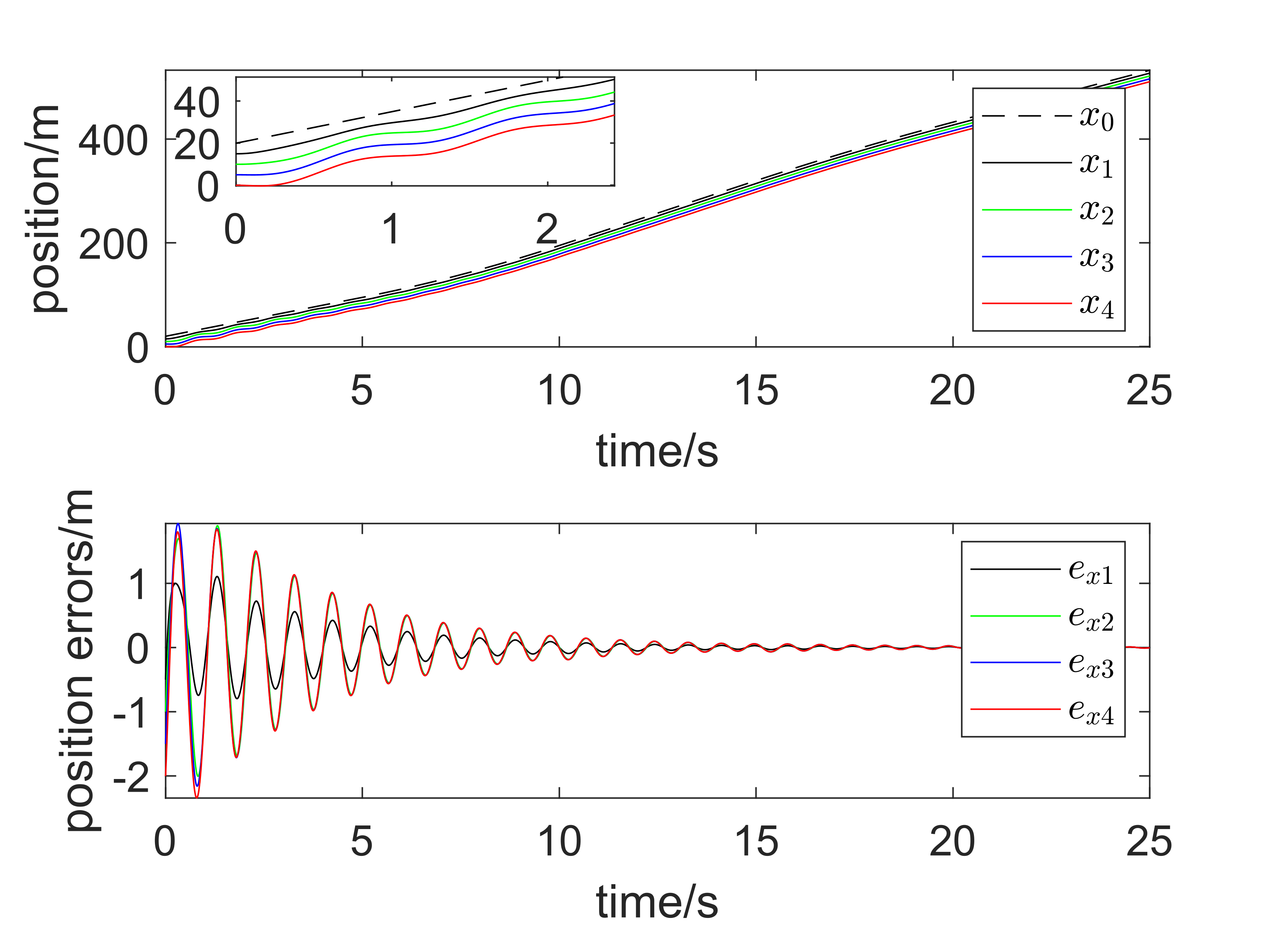}
    \caption{Results observed with the controller proposed in \cite{wang2020finitetime} under sinusoidal disturbances and bidirectional type information flow topology: (Top) Vehicle positions and (Bottom) Position tracking errors.}
    \label{fig:position_position_errors_sin_s_paperB}
\end{figure}
\begin{figure}[!t]
    \centering
    \includegraphics[width=0.8\linewidth]{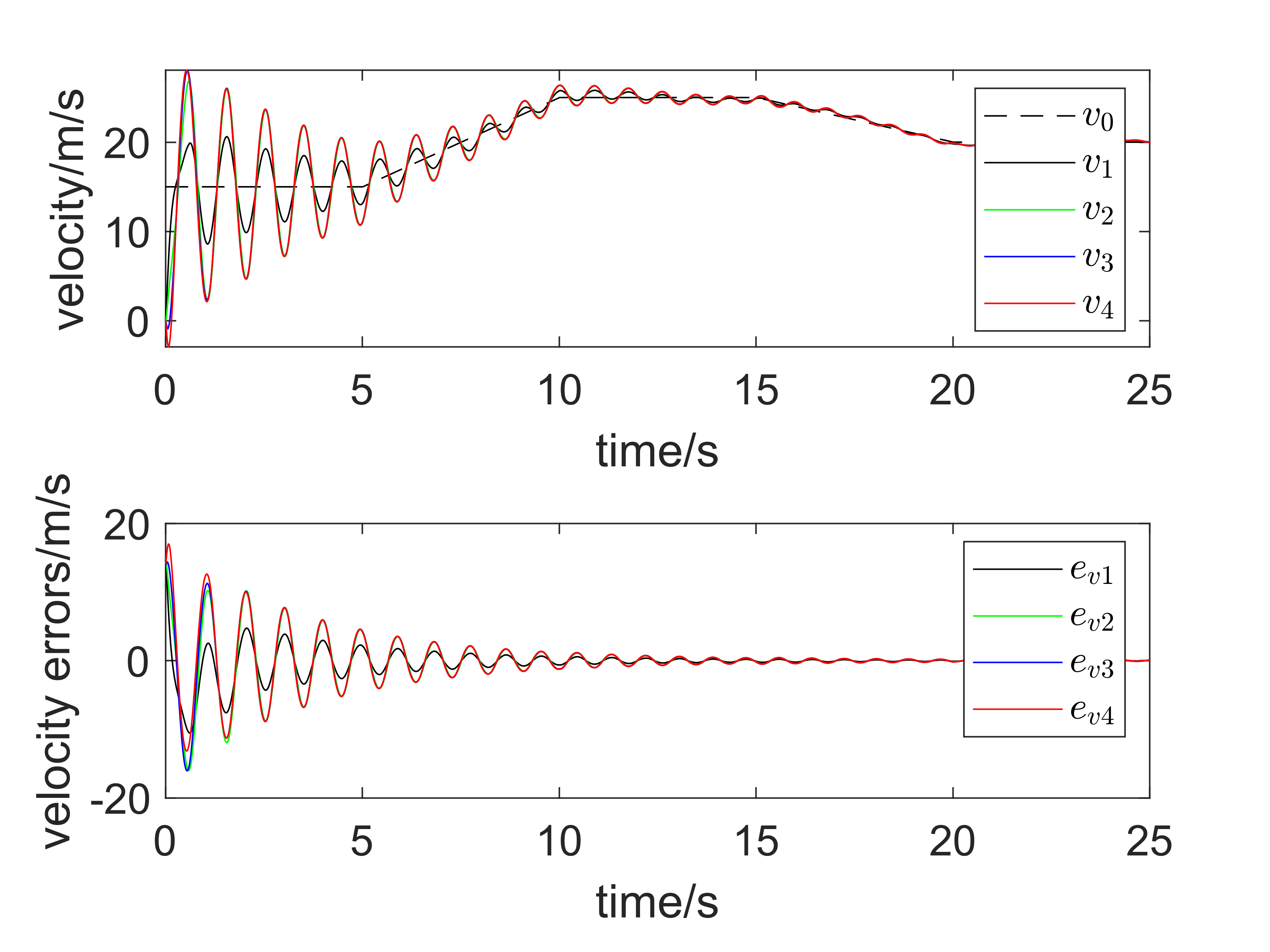}
    \caption{Results observed with the controller proposed in \cite{wang2020finitetime} under sinusoidal disturbances and bidirectional type information flow topology: (Top) Vehicle velocities and (Bottom) Velocity tracking errors.}
    \label{fig:velocity_velocity_errors_sin_s_paperB}
\end{figure}

\begin{figure}[!t]
    \centering
    \includegraphics[width=0.8\linewidth]{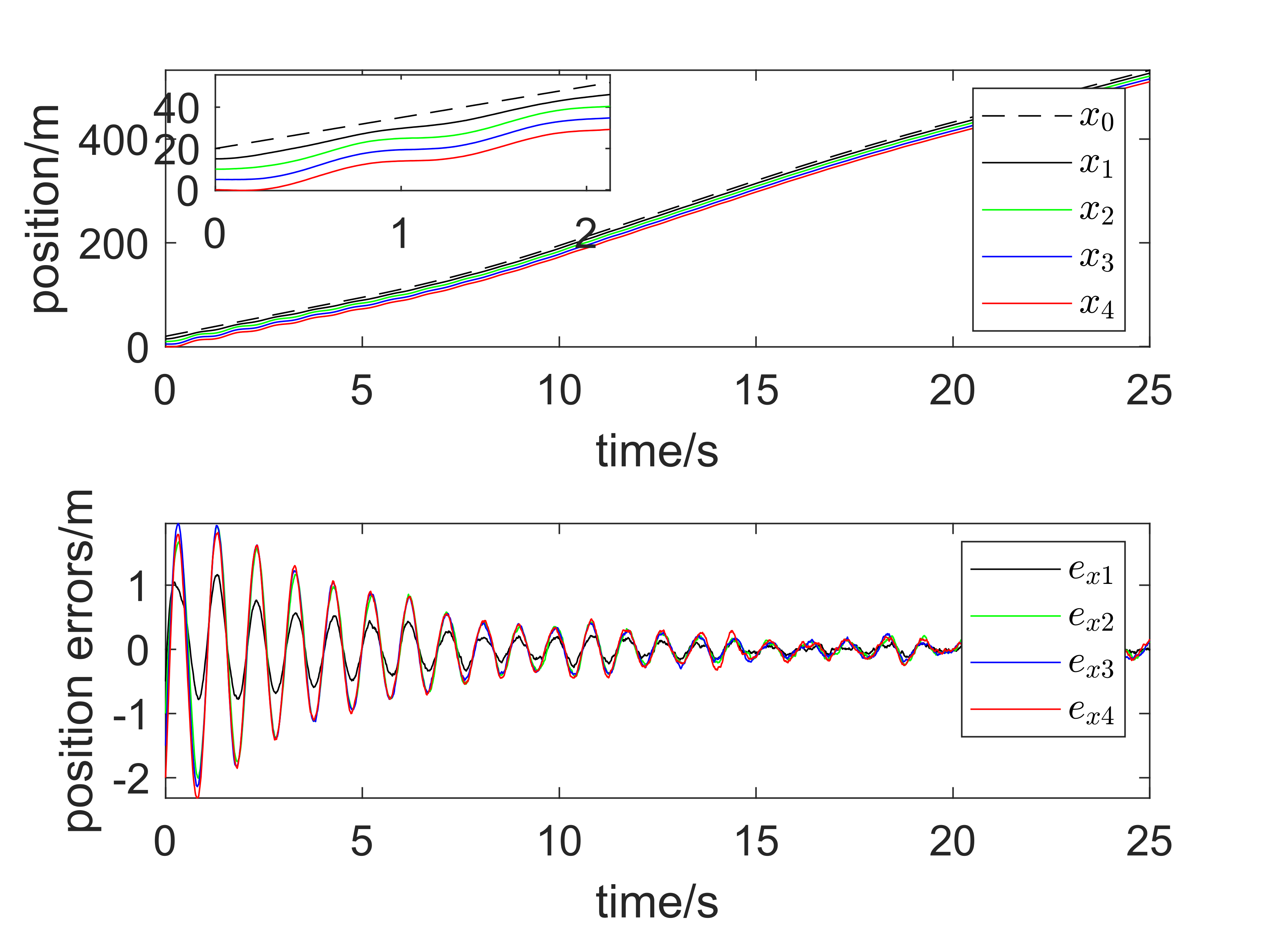}
    \caption{Results observed with the controller proposed in \cite{wang2020finitetime} under Gaussian random noise and bidirectional type information flow topology: (Top) Vehicle positions and (Bottom) Position tracking errors.}
    \label{fig:position_position_errors_aper_s_paperB}
\end{figure}
\begin{figure}[!t]
    \centering
    \includegraphics[width=0.8\linewidth]{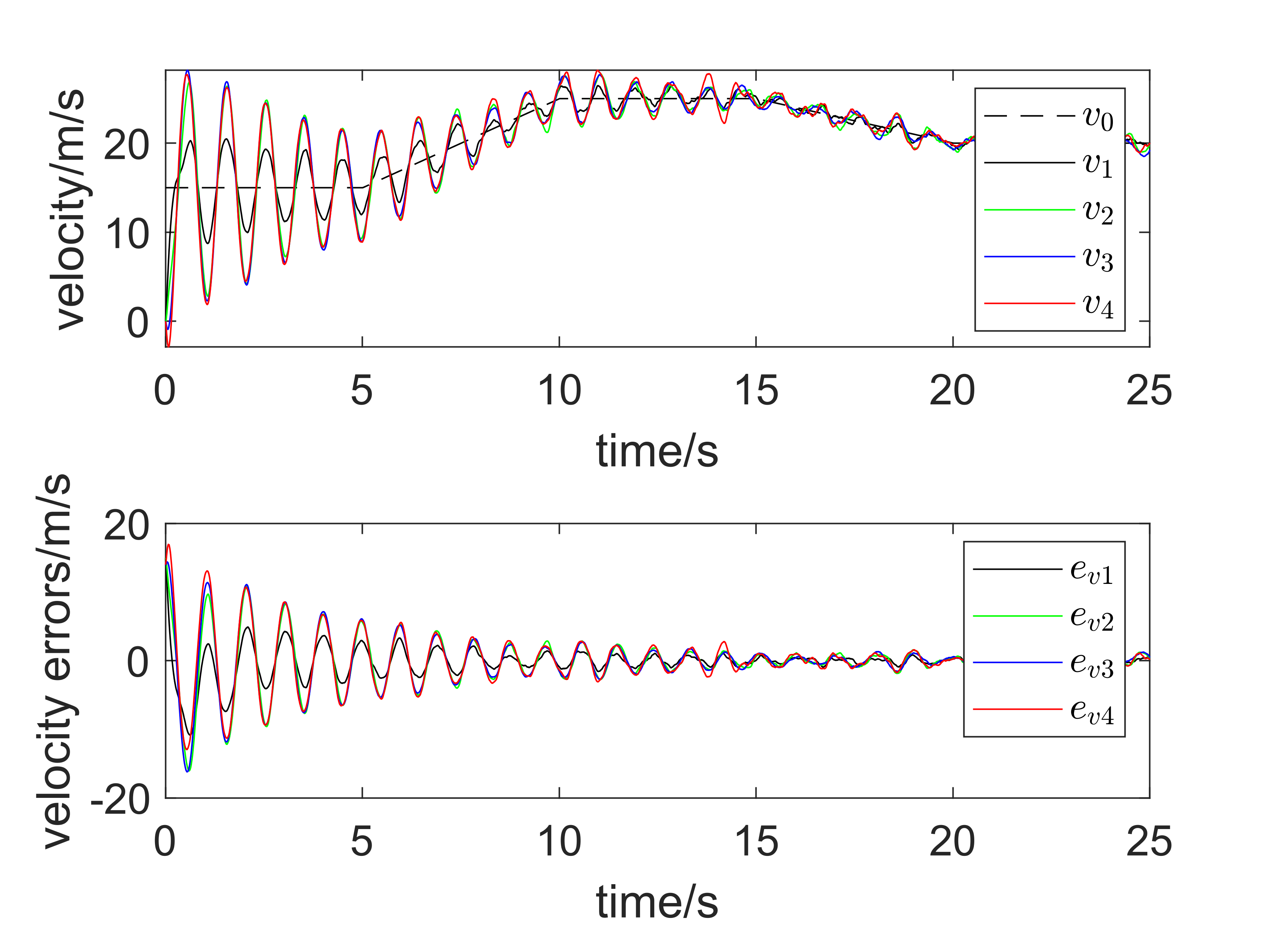}
    \caption{Results observed with the controller proposed in \cite{wang2020finitetime} under Gaussian random noise and bidirectional type information flow topology: (Top) Vehicle velocities and (Bottom) Velocity tracking errors.}
    \label{fig:velocity_velocity_errors_aper_s_paperB}
\end{figure}

\begin{table*}[!t]
\caption{Comparison of the RMS values of the tracking errors for the proposed method and the method proposed in \cite{wang2020finitetime}.}
\centering
\label{tab:comparison_RMS}
\begin{tabular}{|c|cccccccc|cccccccc|}
\hline
Dist. type &
  \multicolumn{8}{c|}{Sinusoidal disturbance} &
  \multicolumn{8}{c|}{Gaussain random noise} \\ \hline
Method &
  \multicolumn{4}{c|}{Our method} &
  \multicolumn{4}{c|}{Method in {\cite{wang2020finitetime}}} &
  \multicolumn{4}{c|}{Our method} &
  \multicolumn{4}{c|}{Method in {\cite{wang2020finitetime}}} \\ \hline
Top. type &
  \multicolumn{2}{c|}{bl-type} &
  \multicolumn{2}{c|}{b-type} &
  \multicolumn{2}{c|}{bl-type} &
  \multicolumn{2}{c|}{b-type} &
  \multicolumn{2}{c|}{bl-type} &
  \multicolumn{2}{c|}{b-type} &
  \multicolumn{2}{c|}{bl-type} &
  \multicolumn{2}{c|}{b-type} \\ \hline
pos./vel. &
  \multicolumn{1}{l|}{pos.} &
  \multicolumn{1}{l|}{vel.} &
  \multicolumn{1}{l|}{pos.} &
  \multicolumn{1}{l|}{vel.} &
  \multicolumn{1}{l|}{pos.} &
  \multicolumn{1}{l|}{vel.} &
  \multicolumn{1}{l|}{pos.} &
  \multicolumn{1}{l|}{vel.} &
  \multicolumn{1}{l|}{pos.} &
  \multicolumn{1}{l|}{vel.} &
  \multicolumn{1}{l|}{pos.} &
  \multicolumn{1}{l|}{vel.} &
  \multicolumn{1}{l|}{pos.} &
  \multicolumn{1}{l|}{vel.} &
  \multicolumn{1}{l|}{pos.} &
  \multicolumn{1}{l|}{vel.} \\ \hline
Veh. 1 &
  \multicolumn{1}{c|}{0.12} &
  \multicolumn{1}{c|}{\textbf{0.76}} &
  \multicolumn{1}{c|}{0.38} &
  \multicolumn{1}{c|}{\textbf{0.79}} &
  \multicolumn{1}{c|}{\textbf{0.09}} &
  \multicolumn{1}{c|}{1.28} &
  \multicolumn{1}{c|}{\textbf{0.24}} &
  1.95 &
  \multicolumn{1}{c|}{\textbf{0.12}} &
  \multicolumn{1}{c|}{\textbf{0.76}} &
  \multicolumn{1}{c|}{0.37} &
  \multicolumn{1}{c|}{\textbf{0.79}} &
  \multicolumn{1}{c|}{\textbf{0.10}} &
  \multicolumn{1}{c|}{1.35} &
  \multicolumn{1}{c|}{\textbf{0.27}} &
  2.10 \\ \hline
Veh. 2 &
  \multicolumn{1}{c|}{\textbf{0.07}} &
  \multicolumn{1}{c|}{\textbf{0.77}} &
  \multicolumn{1}{c|}{0.65} &
  \multicolumn{1}{c|}{\textbf{1.31}} &
  \multicolumn{1}{c|}{0.09} &
  \multicolumn{1}{c|}{1.25} &
  \multicolumn{1}{c|}{\textbf{0.47}} &
  3.23 &
  \multicolumn{1}{c|}{\textbf{0.07}} &
  \multicolumn{1}{c|}{\textbf{0.78}} &
  \multicolumn{1}{c|}{0.64} &
  \multicolumn{1}{c|}{\textbf{1.32}} &
  \multicolumn{1}{c|}{0.10} &
  \multicolumn{1}{c|}{1.32} &
  \multicolumn{1}{c|}{\textbf{0.52}} &
  3.53 \\ \hline
Veh. 3 &
  \multicolumn{1}{c|}{\textbf{0.06}} &
  \multicolumn{1}{c|}{\textbf{0.80}} &
  \multicolumn{1}{c|}{0.81} &
  \multicolumn{1}{c|}{\textbf{1.67}} &
  \multicolumn{1}{c|}{0.09} &
  \multicolumn{1}{c|}{1.19} &
  \multicolumn{1}{c|}{\textbf{0.49}} &
  3.32 &
  \multicolumn{1}{c|}{\textbf{0.06}} &
  \multicolumn{1}{c|}{\textbf{0.80}} &
  \multicolumn{1}{c|}{0.79} &
  \multicolumn{1}{c|}{\textbf{1.68}} &
  \multicolumn{1}{c|}{0.10} &
  \multicolumn{1}{c|}{1.23} &
  \multicolumn{1}{c|}{\textbf{0.54}} &
  3.67 \\ \hline
Veh. 4 &
  \multicolumn{1}{c|}{0.10} &
  \multicolumn{1}{c|}{\textbf{0.83}} &
  \multicolumn{1}{c|}{0.85} &
  \multicolumn{1}{c|}{\textbf{1.87}} &
  \multicolumn{1}{c|}{0.10} &
  \multicolumn{1}{c|}{1.15} &
  \multicolumn{1}{c|}{\textbf{0.49}} &
  3.33 &
  \multicolumn{1}{c|}{\textbf{0.10}} &
  \multicolumn{1}{c|}{\textbf{0.83}} &
  \multicolumn{1}{c|}{0.83} &
  \multicolumn{1}{c|}{\textbf{1.88}} &
  \multicolumn{1}{c|}{0.11} &
  \multicolumn{1}{c|}{1.20} &
  \multicolumn{1}{c|}{\textbf{0.55}} &
  3.69 \\ \hline
\end{tabular}
\end{table*}

\section{Conclusion}\label{sec:conclusion}

In this paper, we proposed a novel notion that we named \emph{Vector String Lyapunov function (VSLF)}, whose existence implies $l_2$ weak string stability. Then, based on this concept, we designed a distributed adaptive backstepping controller for platooning control, which guarantees both compositional and distributed properties. Furthermore, the internal and string stability are formally proved using a VSLF. Simulation results show that our proposed controller is robust with respect to general types of disturbances and it performs better than the existing work \cite{wang2020finitetime} in terms of position and velocity tracking. Future work aims to consider the distributed estimation of the centralized terms $|\mathbf{e}_2|$, $|\mathcal{H}\mathbf{e}_3|$ required in the adaptive laws and some measurement and network induced effects, including sampling-data, delays and noise.


\bibliographystyle{IEEEtran}
\bibliography{references}

\end{document}